\documentclass[journal,twoside,onecolumn,web]{ieeecolor}
\usepackage{generic}
\usepackage{cite}
\usepackage{amsmath,amssymb,amsfonts}
\usepackage{algorithmic}
\usepackage{graphicx}
\usepackage{algorithm,algorithmic}
\usepackage{textcomp}
\usepackage[utf8]{inputenc} 
\usepackage[english]{babel} 
\graphicspath{{images/}} 
\usepackage{mathrsfs}
\usepackage{Macros}
\usepackage[usenames, dvipsnames]{xcolor}

\newtheorem{lemma}{Lemma}

\newtheorem{theorem}[lemma]{Theorem}

\newtheorem{corollary}[lemma]{Corollary}

\newtheorem{remark}[lemma]{Remark}
\newtheorem{definition}[lemma]{Definition}

\def\proofof[#1]{\noindent\hspace{2em}{\itshape #1: }}

\usepackage{tikz}
\usepackage{pgfplots}
\definecolor{customblue}{RGB}{6,61,121}
\pgfplotsset{compat=1.18}
\pgfplotsset{
standard/.style={
			set layers=tick labels on top
		},
layers/tick labels on top/.define layer set=
		{axis background,axis grid,axis ticks,axis lines,main,%
			axis tick labels,
			axis descriptions,axis foreground}
		{/pgfplots/layers/standard}
  }
\usetikzlibrary{cd}
\usetikzlibrary{shapes,fit} 

\usetikzlibrary{decorations.markings}
\tikzset{degil/.style={
		decoration={markings,
			mark= at position 0.5 with {
				\node[transform shape] (tempnode) {$\backslash$};
			}
		},
		postaction={decorate}
	}
}
\allowdisplaybreaks[4]  
\usepackage{ifdraft}
\ifdraft{
    \usepackage{todonotes}      

	\usepackage[normalem]{ulem}
    \usepackage[displaymath, mathlines,switch]{lineno}         
    
}{}

\usepackage[hyphens]{url}					
\usepackage[hidelinks]{hyperref}			
\hypersetup{final=true}

\usepackage{soul} 

\begin{document}
\title{Lyapunov Characterization for ISS of Impulsive Switched Systems}
\author{Saeed Ahmed, Patrick Bachmann, and Stephan Trenn
\thanks{Submitted on May 01, 2024. }
\thanks{Saeed Ahmed is with the Jan C. Willems Center for Systems
and Control and the Engineering, and Technology Institute Groningen (ENTEG), Faculty of Science and Engineering, University of Groningen, 9747 AG Groningen, The Netherlands (e-mail: s.ahmed@rug.nl).  }
\thanks{Patrick Bachmann is with the Institute of Mathematics, University of W{\"u}rzburg, Germany (e-mail: patrick.bachmann@uni-wuerzburg.de).}
\thanks{Stephen Trenn is with the Jan C. Willems Center for Systems
and Control and the  Bernoulli Institute for Mathematics, Computer Science, and Artificial Intelligence, Faculty of Science and Engineering, University of Groningen, The Netherlands (e-mail: s.trenn@rug.nl).}}

\maketitle

\begin{abstract}
In this study, we investigate the ISS of impulsive switched systems that have modes with both stable and unstable flows.  We assume that the switching signal satisfies mode-dependent average dwell and leave time conditions. To establish ISS conditions, we propose two types of time-varying ISS-Lyapunov functions: one that is non-decreasing and another one that is decreasing. Our research proves that the existence of either of these ISS-Lyapunov functions is a necessary and sufficient condition for ISS. We also present a technique for constructing a decreasing ISS-Lyapunov function from a non-decreasing one, which is useful for its own sake. Our findings also have added value to previous research that only studied sufficient conditions for ISS, as our results apply to a broader class of systems. This is because we impose less restrictive dwell and leave time constraints on the switching signal and our ISS-Lyapunov functions are time-varying with general nonlinear conditions imposed on them.  Moreover, we provide a method to guarantee the ISS of a particular class of impulsive switched systems when the switching signal is unknown.
\end{abstract}


\section{Introduction}
Impulsive and switched systems are two important classes of hybrid dynamical systems. Impulsive systems consist of a continuous behavior referred to as a \emph{flow} and abrupt state changes referred to as \emph{jumps}\cite{Simeonov1986}. Switched systems, on the other hand, consist of a family of flows and a switching signal that determines which flow is active at any given time~\cite{liberzon2003switching}. Systems involving impulsive and switching dynamics are ubiquitous in robotics, aircraft, the automotive industry, and network control~\cite{li2005switched}.

When dealing with dynamical systems, it is vital to consider their sensitivity to external inputs or perturbations. Input-to-state stability (ISS) is a useful concept that ensures tolerance to these inputs and helps analyze the system's behavior. The ISS concept can also be applied to analyze the stability and synthesize controllers for dynamical systems with disturbance inputs and complex structures; see, e.g., \cite{jiang1994small}, \cite{heemels2007input}, and \cite{heemels2008input}. However, analyzing the ISS of systems with impulsive and switching dynamics is a challenging problem due to the hybrid nature of these systems. This paper aims to focus on this problem and explore potential solutions.

The stability of switched systems can be categorized based on arbitrary and constrained switching. For a system with arbitrary switching, it is necessary to require that all of its flows are stable. However, even if all of its flows are stable, it is not true in general that the overall switched system is stable. Motivated by this, \cite{dayawansa1999converse} and \cite{lin2009stability} provided several sufficient (and necessary) conditions for the stability of switched systems with arbitrary switching. To ensure the stability of switched systems, whose stability cannot be guaranteed under arbitrary switching, a constraint is imposed on the number of switches via a suitable bound referred to as a \emph{dwell time} constraint. An interesting example of this is in switched systems with stable and unstable flows, whose stability can be guaranteed by quantifying stable and unstable flows in terms of dwell time conditions (\text{cf.} \cite{Ahmed2018, Zhai2001}). A similar dwell time approach is used to guarantee the stability of impulsive systems, quantifying stable flow and unstable jumps, or unstable flow and stable jumps (\text{cf.} \cite{Hespanha2008, Dashkovskiy2013}).

In a similar spirit, we aim at deriving ISS conditions for impulsive switched systems having modes\footnote{It is worth noting that in this paper, the term \emph{mode} has a different definition compared to most switched systems' literature. In switched systems, modes refer to the flows. Here, we expand the definition of the mode to include the succeeding jump as well. Alternatively, one could also consider a flow and its preceding jump as a mode, which is standard in, for example, switched DAEs \cite{Tren12}. Our results can be easily modified to cover this scenario as well.} with stable and unstable flows. Our approach is based on distinguishing between modes with stable and unstable flows and quantifying them to ensure the ISS of the overall impulsive switched system. To accomplish this, we require the switching signal to satisfy two less restrictive constraints: mode-dependent average dwell time (MDADT) and mode-dependent average leave time\footnote{The mode-dependent average leave time is also referred to as mode-dependent reverse-average dwell time in the literature.} (MDALT). MDADT indicates that less frequent jumps can stabilize modes with stable flows, while MDALT suggests that frequent jumps can stabilize modes with unstable flows. 

To establish the ISS conditions, we introduce non-decreasing time-varying ISS-Lyapunov functions and decreasing time-varying ISS-Lyapunov functions, which are a subclass of the non-decreasing ones. Then, we show that the existence of the non-decreasing ones gives a sufficient condition for ISS and the decreasing ones provide a necessary condition. Therefore, it implies that the existence of both non-decreasing and decreasing ISS Lyapunov functions is equivalent to ISS. 

To the best of the authors' knowledge,  \emph{a converse ISS-Lyapunov theorem} has not been provided before for impulsive switched systems. Moreover, we provide a technique for \emph{constructing a decreasing ISS-Lyapunov function from a non-decreasing one}, which is important and useful in its own right. This construction is motivated by the facts that (i) the relation between decreasing and non-decreasing ISS-Lyapunov functions is important because every decreasing Lyapunov function is also non-decreasing but the converse does not hold true in general, 
and (ii) the level sets of a decreasing ISS-Lyapunov function directly indicate the reachable sets. Our results are summarized in the Fig.~\ref{fig:ISSLFEquiv}.

\begin{figure}[htbp]
    \vspace{.3cm}
    \centering
    \begin{tikzpicture}
        \node (ISS) at (0,0) {ISS};
        \node (nondecrLF) at (-2.8,-2) [text width=3.3cm,align=center]{$\exists$ a non-decreasing ISS-Lyapunov function};
        \node (decrLF) at (2.8,-2) [text width=3cm]{$\exists$ a decreasing ISS-Lyapunov function};
        
        \path
        (ISS) edge[thick,double,double equal sign distance,-{Implies[]}] node[above right=.1cm]{\footnotesize Thm. \ref{thm:converseLyapunovTheorem}}(decrLF)
        ([shift={(0cm,-0.2cm)}]decrLF.west) edge[thick,double,double equal sign distance,-{Implies[]}]
        node[below=.1cm]{\footnotesize Def. \ref{def:decrLF}} ([shift={(0cm,-0.2cm)}]nondecrLF.east)
        ([shift={(0cm,0.2cm)}]decrLF.west) edge[thick,double,double equal sign distance,{Implies[]}-]
        node[above=.1cm]{\footnotesize Thm. \ref{thm:ConstructionOfLF}} ([shift={(0cm,0.2cm)}]nondecrLF.east)
        (nondecrLF) edge[thick,double,double equal sign distance,-{Implies[]}] node[above left=.1cm]{\footnotesize Thm. \ref{thm:ISS}} (ISS);
    \end{tikzpicture}
    \caption{Summary of our results}
    \label{fig:ISSLFEquiv}
\end{figure}
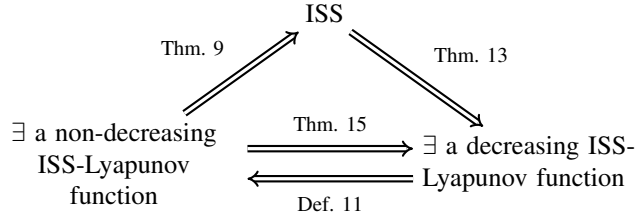

The results available in the literature for ISS of impulsive switched systems 
\cite{liu2012class, li2018input, zhu2020input, mancilla2020uniform}
 only provide a \emph{sufficient} condition of ISS. We even have added value to these results, which are summarized below (cf.\ our recent conference paper \cite{Saeed2024atime-varying}):
\begin{itemize}
\item Our results apply to systems in which some modes have stable flows, while others may have unstable flows. However, the results of \cite{liu2012class} and \cite{mancilla2020uniform} are limited to systems in which all modes have either stable or unstable flows.
    \item Our approach utilizes time-varying ISS-Lyapunov functions and thus can provide ISS conclusions for impulsive switched systems with simultaneous instability of flow and jumps, while results of \cite{li2018input} and \cite{zhu2020input} cannot be used to conclude the ISS of this class of systems.
    \item Our dwell and leave time constraints are mode-dependent and we do not require them to be the same as in \cite{li2018input}. This allows us to consider a broader class of systems as compared to \cite{li2018input}.
    \item We provide an approach to achieve robustness of our ISS results with respect to unknown switching signals for a class of impulsive switched systems with time-independent flow and jump maps.
    \item Our approach restricted to linear systems is constructive in the sense that a set of LMIs can be defined, whose feasibility guarantees ISS for a class of unknown switchings.
    \item We complement the approach in \cite{Yin2023} by providing ISS and including impulsive dynamics and infinitely many modes.
\end{itemize}

This work is an extension of our conference paper \cite{Saeed2024atime-varying}. It provides several enhancements, including a detailed proof of Theorem~\ref{thm:ISS}, which was only sketched in the conference version. Additionally, it allows an infinite number of modes, whereas the conference version was restricted to a finite number of modes. It also introduces a necessary and sufficient condition for ISS through the introduction of a decreasing ISS-Lyapunov function, whereas only a sufficient condition of ISS was provided in the conference version via a non-decreasing ISS-Lyapunov function. Moreover, it presents a method to construct a decreasing ISS-Lyapunov function from a non-decreasing one. Finally, it provides a method to guarantee ISS of a class of impulsive switched systems with time-independent jump and flow maps when the switching signal is unknown. We additionally apply the aforementioned result to linear systems to obtain a set of LMIs guaranteeing ISS.

The rest of the paper is organized as follows. In Section~\ref{sec:prelim}, we will introduce notation, system description, necessary definitions, and problem formulation. In Section~\ref{sec:ISS_result}, we will provide a sufficient condition for ISS of the considered class of systems via a non-decreasing ISS-Lyapunov function. In Section~\ref{sec:nece_and_suff}, we will provide a necessary and sufficient condition of ISS via a decreasing ISS-Lyapunov function. In Section~\ref{sec:construction}, we will suggest a method for constructing a decreasing ISS-Lyapunov function from the non-decreasing ISS-Lyapunov function proposed in Section~\ref{sec:ISS_result}. In Section~\ref{sec:robustness}, we will provide a method to guarantee ISS for impulsive switched systems with time-independent flow and jump maps when the switching signal is unknown. We will then apply this robustness result to find sufficient conditions of ISS for linear systems in terms of LMIs in the special case when the ISS-Lyapunov functions are quadratic and time-independent. Finally, we will conclude the paper by summarizing our work in Section~\ref{sec:conclusion}.

\section{Preliminaries and Problem Formulation}
\label{sec:prelim}
We denote the set of positive integers by $\N$, the set of nonnegative integers by $\N_0$, the set of real numbers by $\R$,  the set of nonnegative real numbers by $\R^+_0$, the space of continuous functions from normed spaces $X$ to $Y$ by $\Cont(X,Y)$, the space of locally bounded piecewise continuous functions from $X$ to $Y$ by $\mathcal{PC}\!\para{X, Y}$, the identity function $\id: X \to X$, $x \mapsto x$, and  the ball of radius $r > 0$ around $0$ in $X$ by $\ball_X(r)$. Let $X \subseteq \R^n$ and $U \subseteq \R^m$ represent the state space and input space, respectively. Let $U_c$ be the space of bounded functions from $[0, \infty)$ to $U$ with norm $\norm{u}_\infty := \sup_{t \in [0, \infty)}\!\braces{\norm{u(t)}_U}$. For a function $f\colon X \to Y$, we denote the image by $\image(f)$.
We denote the left (right) limit of a function $f$ at $t$ as $f(t^-)$ ($f(t^+)$), and implicitly assume that the limit is well-defined when using this notation. For a continuous function $V: C \to \R$, $C \subseteq \R$, we denote the (upper) \emph{Dini-derivative} by
\begin{align*}
	&\tfrac{\diff}{\diff t}V(t)
	= \limsup\limits_{s \searrow 0} \tfrac{1}{s}(V(t+s)) - V(t).
\end{align*}

To introduce the notion of ISS-Lyapunov functions and corresponding ISS results later, we define the  \emph{comparison function classes} as follows: \emph{Class} $\mathcal{P}$ is the set of all continuous functions $\gamma:[0,\infty)\to[0,\infty)$, which satisfy $\gamma(0)=0$ and $\gamma(r)>0$ for all $r>0$. \emph{Class} $\mathcal{K}$ is the set of all continuous functions $\gamma:[0,\infty)\to[0,\infty)$, which are strictly increasing and $\gamma(0)=0$. \emph{Class} $\mathcal{K}_\infty$ is the subset of class $\mathcal{K}$ for which additionally $\gamma(s)\to\infty$ as $s\to\infty$. \emph{Class} $\mathcal{KL}$ is the set of all continuous functions $\beta:[0,\infty) \times [0,\infty) \to [0,\infty)$ for which $\beta(\ph,r)$ is class $\mathcal{K}$ for every fixed $r \geq 0$, and for each fixed $s > 0$, the mapping $r\mapsto\beta(s,r)$ is strictly decreasing and converges to zero as $r\to\infty$. 

\begin{definition}
    Let $t_0 \in \R$. A \emph{switching signal} is a piecewise-constant, left-continuous function $\sigma:[t_0,\infty)\to\mathcal{M}$, where $\mathcal{M}$ is some (finite or infinite) index set. The set $S = \braces{t_i}_{i \in \N} \subset (t_0,\infty)$ of discontinuities of $\sigma$ is called the set of \emph{switching instants} and it is assumed that the sequence $(t_i)_{i\in\N}$ is strictly increasing and unbounded, i.e., no accumulation towards a finite time (so-called Zeno behavior) is considered. The sequence $(p_i)_{i\in\N_0}$ such that $\sigma(t)=p_i$ on $[t_i,t_{i+1})$ is called \emph{mode sequence} of $\sigma$. For the interval $I := [t_0,\infty)$ and $p \in \mathcal M$, we denote by $I^{\sigma}_p$ the subset of $I$, on which mode $p$ of a given switching signal $\sigma$ is active, i.e.\ $I^{\sigma}_p:=\{t\in I\,\mid\, \sigma(t)=p\}$.
\end{definition}

In this paper, we consider \emph{impulsive switched systems} of the form:
\begin{align*}
    \Sigma\colon \left\{
    \begin{aligned}
        \dot{x}(t) &= f_{\sigma(t)}\!\para{t,x(t),u(t)}, \quad &&t \in I\setminus S, \\
        x(t_i^+) &=g_{\sigma(t_i^-)}\!\para{t_i^-,x(t_i^-),u(t_i^-)}, \quad &&t_i\in S,
    \end{aligned}
    \right.
\end{align*}
where $I = [t_0,\infty)$, $\sigma$ is a switching signal with corresponding set of switching instants $S$, $x\colon I \to X$, and $u \in U_c$. We assume that for every $D > 0$, functions $f_p,g_p\colon I\times X \times U\to X$ are locally Lipschitz continuous in the second argument, uniformly for all $t \in I$, $p \in \mathcal M$ and $u \in B_{U_c}(D)$. We call $x:I \to X$ a solution of $\Sigma$ for some given $\sigma$ and some input $u\in U_c$, if $x$ is locally absolutely continuous on $I\setminus S$ with well defined left- and right-limits at all $t_i\in S$ such that the equations of $\Sigma$ hold for almost all $t\in I\setminus S$ and all $t_i\in S$. Without loss of generality, we assume that every solution $x$ is right continuous and hence $x(t_i)=x(t_i^+)$. We furthermore assume that $f_p$ and $g_p$ are such that  $\Sigma$ is \emph{robustly forward complete} (or has bounded reachability sets), i.e., for all initial conditions $(t_0,x_0) \in \R \times X
$, $C > 0$, $D > 0$ and $\tau > 0$,
\begin{align}\label{eq:forwardCompleteness}
    K(C, D, \tau, \sigma) 
    := \hspace{-0.5mm}\sup_{x_0 \in B_X(C),\, u \in B_{U_c}(D),\, t \in [t_0,t_0 + \tau]} \!\! \norm{x(t;t_0, x_0, u, \sigma)}_X
\end{align}
is finite, where $x(t;t_0,x_0,u,\sigma)$ denotes the (unique) solution of $\Sigma$. We define a \emph{mode} of $\Sigma$ as a flow between two consecutive switching instants and the succeeding jump.

In this work, we want to study input-to-state stability (ISS) of $\Sigma$, which is formally defined as follows.
\begin{definition}\label{def:ISS}
	For a given switching sequence $\sigma$, we call system $\Sigma$ \emph{input-to-state stable (ISS)} if there exist functions $\beta \in \mathcal{KL}$ and $\gamma \in \mathcal{K}_\infty$ such that for all initial conditions $x_0 \in X$ and every input function $u \in U_c$, the system has a global solution, which satisfies for all $t \in [t_0, \infty)$,
	\begin{align*}
		\norm{x(t;t_0,x_0,u,\sigma)}_X \leq \beta\!\para{\norm{x_0}_X, t - t_0} + \gamma\!\para{\norm{u}_\infty}.
	\end{align*}
\end{definition}

For a given switching signal $\sigma$, let $N^\sigma_{p}(s_1,s_2)$ denote the number of times that mode $p$ is activated in the interval $(s_1,s_2]$ and $T^\sigma_p(s_1,s_2)$ denote the time that mode $p$ is active in interval $[s_1,s_2)$. 

We conclude this section by defining MDADT and MDALT for a switching signal; for this we assume that the mode set $\mathcal{M}$ is composed of two disjoint subsets  $\mathcal{S},\mathcal{U}\subseteq\mathcal{M}$ and we define MDADT only for modes in $\mathcal{S}$ and MDALT for modes in $\mathcal{U}$. For the definition of MDADT and MDALT, the decomposition of $\mathcal{M}$ can be arbitrary, but we will later choose a decomposition based on the stability\footnote{Stability here is with respect to a chosen, possibly time-varying, candidate ISS Lyapunov function; in particular, for a different choice of candidate ISS Lyapunov function, the stabilty classification of modes may be different.}  of the flows; see Definition~\ref{def:nonDecrLF}.

\begin{definition}\label{def:MDADT}
    Consider the modes $\mathcal S \subseteq \mathcal M$ and let $\{\tau_p\}_{p \in \mathcal S}, \tau_p \geq \tau > 0$. If for the switching signal $\sigma \colon I \to \mathcal M$, there exists a constant $T_{\mathcal S} \geq 0$, such that for all $s_1, s_2 \in I$, $s_1 \leq s_2$, the inequality
\begin{align}\label{ineq:MDADT}
        \sum_{p \in \mathcal S} N^{\sigma}_p(s_1,s_2)\tau_p - T^\sigma_p(s_1,s_2) \leq T_{\mathcal S}
    \end{align}
   holds true, then we say that $\sigma$ has MDADT $\{\tau_p\}_{p \in \mathcal S}$ (for all modes in $\mathcal{S}$), or short, $\sigma$ satisfies the MDADT condition \eqref{ineq:MDADT}.
\end{definition}

\begin{definition}\label{def:MDALT}
    Consider the modes $\mathcal U \subseteq \mathcal M$ and let $\{\tau_p\}_{p \in \mathcal U}, \tau_p \geq \tau > 0$. If for the switching signal $\sigma \colon I \to \mathcal M$, there exists a constant $T_{\mathcal U} \geq 0$, such that for all $s_1, s_2 \in I$, $s_1 \leq s_2$, the inequality
    \begin{align}\label{ineq:MDALT}
        \sum_{p \in \mathcal U} N^{\sigma}_p(s_1,s_2)\tau_p - T^\sigma_p(s_1,s_2) \geq  -T_{\mathcal U}
    \end{align}
    holds true, then we say that $\sigma$ has the  MDALT $\{\tau_p\}_{p \in \mathcal U}$ (for all modes in $\mathcal{U}$), or short, $\sigma$ satisfies the MDALT condition \eqref{ineq:MDALT}.
\end{definition}

\begin{remark}
    We highlight here that Definitions \ref{def:MDADT} and \ref{def:MDALT} are formulated to allow switched systems with an infinite number of modes, which is one of our contributions. For switched systems with finitely many modes and without impulsive effects, these definitions reduce to the classical dwell time/leave time conditions as introduced in  Definitions 4.1 and 4.2 of \cite{Yin2023}.
\end{remark}



\section{Sufficient condition for ISS}
\label{sec:ISS_result}
In this section, we provide one of our main results on ISS of system~$\Sigma$. But before we proceed, we provide the notion of a candidate ISS-Lyapunov functions as follows:

\begin{definition}\label{def:candLyapunovFunction}
    Consider the system $\Sigma$ with some given switching signal $\sigma$ and let $\widetilde V_p \in \Cont(I^{\sigma}_p \times X, \R^+_0)$ for $p \in \mathcal{M}$.
     We call $V_\sigma \in \mathcal{PC}(I \times X,\R_0^+)$ given by $V_\sigma(t,x) := \widetilde V_{\sigma(t)}(t,x)$ a \emph{candidate ISS-Lyapunov function (in implication form)} for system $\Sigma$ with switching signal $\sigma$, if it fulfills all of the following conditions:
	\begin{enumerate}
		\item \label{def:LyapFuncCondition1} There exist functions $\alpha_1, \alpha_2 \in \mathcal{K}_\infty$ such that
		\begin{align}
			\alpha_1\!\para{\norm{x}_X} \leq V_\sigma(t,x) \leq \alpha_2\!\para{\norm{x}_X} \label{ineq:LyapunovDefiniteness}
		\end{align}
		holds true for all $t \in I$ and all $x \in X$. 
		\item \label{def:LyapFuncCondition2} There exist functions $\chi \in \mathcal{K}_\infty$ and $\underline \varphi, \overline \varphi \in \PP$ such that for each $p \in \mathcal M$, there exist $\psi_p \in \PP$ and $\varphi_p \in \PP \cup - \PP$ with $\underline \varphi(s) \leq\abs{\varphi_p(s)} \leq \overline \varphi(s)$ $\forall s \in \R^+_0$ such that for all inputs $u \in U_c$ and all solutions $x(t) = x(t; t_0, x_0, u, \sigma)$ of $\Sigma$, the inequality
		\begin{align}
			&\tfrac{\diff}{\diff t}V_\sigma(t, x(t)) \leq \varphi_{\sigma(t)}\!\para{V_\sigma(t,x(t))}, \ t \in I \setminus S, \label{ineq:LyapunovFlow}
        \end{align}
        holds whenever $V_\sigma(t, x(t)) \geq \chi\!\para{\norm{u}_\infty}$, and
        \begin{align}
			&V_\sigma(t_i,x_i^+) \leq \psi_{\sigma(t_i^-)}\!\para{V_\sigma(t_i^-,x(t_i^-))},\ t_i \in S, \label{ineq:LyapunovJump}
		\end{align}
		holds, whenever $V_\sigma(t_i^-, x(t_i^-)) \geq \chi\!\para{\norm{u}_\infty}$, where $x_i^+=g_{\sigma(t_i^-)}(t_i^-,x(t_i^-), u(t_i^-))$.
		\item \label{def:LyapFuncCondition3} There exists a function $\alpha_3 \in \mathcal{K}$, such that for all $x \in X$, all $u \in U_c$, and all $i \in \N$, which satisfy $V_\sigma\!\para{t_i^-,x} < \chi\!\para{\norm{u}_\infty}$, the jump rule satisfies
		\begin{align}
	    	V_\sigma\!\para{t_i,g_{\sigma(t_i^-)}\!\para{t_i^-,x, u(t_i^-)}} &\leq \alpha_3\!\para{\norm{u}_\infty}. \label{ineq:LyapFuncCondition3}
		\end{align}
	\end{enumerate}
    Furthermore, we call $V_\sigma$ a \textit{candidate ISS-Lyapunov function in dissipation form} if \ref{def:LyapFuncCondition2}) is replaced by
    \begin{enumerate}
        \item[2')] There exist functions $\chi \in \mathcal{K}_\infty$ and $\underline \varphi, \overline \varphi \in \PP$, such that for each $p \in \mathcal M$, there exist $\psi_p \in \PP$ and $\varphi_p \in \PP \cup - \PP$ with $\underline \varphi(s) \leq\abs{\varphi_p(s)} \leq \overline \varphi(s)$ $\forall s \in \R^+_0$, such that for all inputs $u \in U_c$ and all solutions $x(t) = x(t; t_0, x_0, u, \sigma)$ of $\Sigma$, the inequalities
        \begin{align}
            \tfrac{\diff}{\diff t}V_\sigma(t, x(t)) &\leq \varphi_{\sigma(t)}\!\para{V_\sigma(t,x(t))} + \chi\!\para{\norm{u}_\infty},\  t \in I \setminus S, \label{ineq:LyapunovDissipationFlow} \\
            V_\sigma(t_i,x_i^+) &\leq \psi_{\sigma(t_i^-)}\!\para{V_\sigma(t_i^-,x(t_i^-))} + \chi\!\para{\norm{u}_\infty}, \ t_i \in S, \label{ineq:LyapunovDissipationJump}
        \end{align}
        hold true, where $x_i^+=g_{\sigma(t_i^-)}(t_i^-,x(t_i^-), u(t_i^-))$.
    \end{enumerate}
\end{definition}

Let us define the functions $\Phi_p, \underline \Phi, \overline \Phi\colon [0,\infty) \rightarrow \R \cup \{-\infty\}$ for $p \in \mathcal M$ as
\begin{align*}
 \Phi_p(v)& := \int_{1}^{v} \frac{1}{\abs{\varphi_p(s)}} \ \diff s, \quad
    \underline \Phi(v) := \int_{1}^{v} \frac{1}{\underline \varphi(s)} \ \diff s,\quad   
    \overline \Phi(v) := \int_{1}^{v} \frac{1}{\overline \varphi(s)} \ \diff s.
\end{align*}
Note that the functions $\Phi_p, \underline \Phi, \overline \Phi$ are all strictly increasing.
Therefore, their inverses $\Phi^{-1}_p: \image(\Phi_p) \rightarrow [0,\infty)$, $\underline \Phi^{-1}\colon \image(\underline \Phi) \rightarrow [0,\infty)$ and $\overline \Phi^{-1}\colon \image(\overline \Phi) \rightarrow [0,\infty)$ exist. Furthermore, note that if \eqref{ineq:LyapunovFlow} holds with a linear bound, i.e.\ $\tfrac{\diff}{\diff t}V_{\sigma}\leq \lambda_\sigma V_\sigma$, then $\Phi_p(v) = \frac{\ln v}{|\lambda_p|}$.

\begin{definition}\label{def:nonDecrLF}
    Let $\Sigma$ be a switched system with switching signal $\sigma$. Let $V_\sigma$ be a 
    candidate ISS-Lyapunov function for system $\Sigma$ with corresponding functions $\varphi_p, \psi_p$ as in Definition~\ref{def:candLyapunovFunction} and let $\mathcal{M} = \mathcal{S} \,\dot{\cup}\, \mathcal{U}$, such that $-\varphi_p \in \mathcal P$ for $p \in \mathcal S$ and $\varphi_p\in\mathcal{P}$ for all $p\in\mathcal{U}$. Furthermore, assume that $\sigma$ satisfies the MDADT condition \eqref{ineq:MDADT} for all modes in $\mathcal{S}$ and the MDALT condition \eqref{ineq:MDALT} for all modes in $\mathcal{U}$ with corresponding dwell/leave times $\{\tau_p\}_{p \in \mathcal{S}}$ and $\{\tau_p\}_{p \in \mathcal{U}}$. If there exists $\delta > 0$ such that for all $a > 0$ and every switching time $t_i \in S$, one of the following two inequalities is satisfied:
    \begin{align} \label{ineq:dwellTimeConditionSFUJ}
        \Phi_{\sigma(t_i)}(\psi_{\sigma(t_i^-)}(a)) - \Phi_{\sigma(t_i^-)}(a) \leq \tau_{\sigma(t_i^-)}(1 - \delta),
    \end{align}
    if $\sigma(t_i^-)\in\mathcal{S}$, or
    \begin{align} \label{ineq:dwellTimeConditionUFSJ}
        -\Phi_{\sigma(t_i)}(\psi_{\sigma(t_i^-)}(a)) + \Phi_{\sigma(t_i^-)}(a) \geq \tau_{\sigma(t_i^-)}(1 + \delta),
    \end{align}
    if $\sigma(t_i^-)\in\mathcal{U}$, then we call $V_\sigma$ a \emph{non-decreasing ISS-Lyapunov function}.
\end{definition}

\begin{remark}
    We emphasize here that with ``non-decreasing'' we mean ``not necessarily decreasing'', i.e.\ in contrast to classical Lyapunov functions, the value of the non-decreasing Lyapunov function is allowed to increase along a solution trajectory (cf.~\cite{DefoDjem14}).   
\end{remark}

Now we are ready to provide our result on ISS of impulsive switched system $\Sigma$.
\begin{theorem}\label{thm:ISS}
    Consider the system $\Sigma$ for a given switching signal $\sigma$. If there exists a non-decreasing ISS-Lyapunov function as given in Definition \ref{def:nonDecrLF},
    then $\Sigma$ is ISS.
\end{theorem}

\begin{proof}
    For an arbitrary but fixed input signal $u\in U_c$, consider the set 
\begin{align*}
    A_1(t) := \braces{x \in X \, \middle| \, V_\sigma(t,x) < \chi\!\para{\norm{u}_{\infty}}}.
\end{align*}
Our proof consists of two steps.
At first, we show that for all initial conditions $x_0$ outside $A_1(t_0)$, we have a convergent behavior towards $A_1(t)$, i.e., there exists a $\mathcal{KL}$-function $\beta$, such that the inequality
\begin{equation}\label{eq:beta_proof}
    {\norm{x(t;t_0,x_0,u, \sigma)}_X \leq \beta\!\para{\norm{x_0}_X\!, t - t_0}}
\end{equation}
holds true for all $t\in[t_0,t_*)$, where $t_*>t_0$ is  such that $x(t) \notin A_1(t)$ for all $t \in [t_0,t_*)$. Second, we show that trajectories, once they have reached the set $A_1$, will stay bounded.

\textbf{Step 1}: Let $t_* = \inf\!\braces{t \in [t_0,\infty] \, \middle| \, x(t) \in A_1(t)}$, i.e., $V_\sigma(t,x(t)) \geq \chi\!\para{\norm{u}_\infty}$ holds for all $t \in [t_0, t_*)$ and hence \eqref{ineq:LyapunovFlow} and \eqref{ineq:LyapunovJump} are satisfied on $[t_0,t_*)$.
For brevity, define $v(t) := V_\sigma\!\para{t,x(t)}$ and denote $v_i := v(t_i)$ and $v^-_i := v(t_i^-)$. Let $t'=t^*$ if $v(t)\neq 0$ for all $t\in[t_0,t^*)$ or $t' = \min \{t\geq t_0\,|\,v(t)=0\}$ otherwise. In the following, we will show that \eqref{eq:beta_proof} holds for all $t\in[t_0,t')$. If $t'<t_*$, then it follows from positive definiteness of $V_\sigma$ that $x(t')=0$. Furthermore, from $0=v(t')=V_\sigma(t',x(t'))\geq \xi(\|u\|_\infty)$, it follows that $u$ must be identically zero, which then implies that $x(t)=0$ for all $t\geq t'$. Consequently, any extension of a $\mathcal{KL}$-function $\beta$ in \eqref{eq:beta_proof} onto the interval $[t_0,t_*)$ also makes \eqref{eq:beta_proof} true on the whole interval $[t_0,t_*)$.
Now we consider, the behavior on $[t_0,t')$ on which $v(t)\neq 0$. Then for $p \in \mathcal M$, the inequality \eqref{ineq:LyapunovFlow} becomes
\begin{align}\label{ineq:transformIDEforV_flow}
    \frac{\frac{\diff}{\diff t}v(t)}{\abs{\varphi_p\!\para{v(t)}}}
    \leq\frac{\varphi_p\!\para{v(t)}}{\abs{\varphi_p\!\para{v(t)}}} = 
    \begin{cases}
        1,  &\text{if } \varphi_p\!\para{v(t)} > 0, \\
        -1, &\text{if } \varphi_p\!\para{v(t)} < 0,
    \end{cases}
\end{align}
for all $t\in I_p^\sigma\cap [t_0,t')$.

Now, we will estimate $\Phi_{\sigma(t)}(v(t)) - \Phi_{\sigma(0)}(v_0)$ to conclude that it is bounded by a class $\mathcal{KL}$ function. To this aim, we first estimate the behavior between the switching instants and at the switching instants, separately.

Integrating \eqref{ineq:transformIDEforV_flow} over the interval $[t_i,\hat{t}]$, $i \in \N_0$, for some $\hat{t} \in [t_i,t_{i+1}) \cap [t_0,t')$, we obtain
\begin{align*}
    \int_{v_i}^{v(\hat{t})} \frac{1}{\abs{\varphi_p(s)}}\ \diff s
    &= \int_{t_i}^{\hat{t}} \frac{\frac{\diff}{\diff t} v(t)}{\abs{\varphi_p\!\para{v(t)}}}\ \diff t
    \leq 
    \begin{cases}
        -(\hat{t} - t_i), &\text{if } p \in \mathcal S, \\
        \hat{t} - t_i,  &\text{if } p \in \mathcal U,
    \end{cases}
\end{align*}
where we used the integration parameter change $s \to v(t)$. It follows that
\begin{align}
    \Phi_p(v(\hat{t})) - \Phi_p(v_i) 
    &\leq 
    \begin{cases}
        -T^\sigma_p(t_i,\hat{t}), &\text{if } p \in \mathcal S, \\
        T^\sigma_p(t_i,\hat{t}),  &\text{if } p \in \mathcal U.
    \end{cases}
    \label{ineq:estimateFlowVariableUpperBound}
\end{align}

For the switching instants $t_i \in [t_0,t')$, the inequality
\begin{align}
    &\Phi_{\sigma(t_i)}(v_i)  -  \Phi_{\sigma(t^-_i)}(v^-_i) \leq \Phi_{\sigma(t_i)}(\psi_{\sigma(t_i^-)}\para{v^-_i})  -  \Phi_{\sigma(t^-_i)}(v^-_i) 
    \leq \tau_{\sigma(t_i^-)}(1 - \delta) \label{ineq:estimateswitchingVariableUpperBoundSFUJ}
\end{align}
holds for $\sigma(t_i^-) \in \mathcal S$. Here, we used \eqref{ineq:LyapunovJump} in the first inequality and \eqref{ineq:dwellTimeConditionSFUJ} in the second. Analogously, from \eqref{ineq:dwellTimeConditionUFSJ}, we obtain
\begin{align}
    &\Phi_{\sigma(t_i)}(v_i)  -  \Phi_{\sigma(t^-_i)}(v^-_i) \leq \Phi_{\sigma(t_i)}(\psi_{\sigma(t_i^-)}\para{v^-_i})  -  \Phi_{\sigma(t^-_i)}(v^-_i) 
    \leq -\tau_{\sigma(t_i^-)}(1 + \delta) \label{ineq:estimateswitchingVariableUpperBoundUFSJ}
\end{align}
for $\sigma(t_i^-) \in \mathcal U$.

Let $n = n(t) := \sum_{p \in M} N^{\sigma}_p(t_0,t)$.
With the estimates \eqref{ineq:estimateFlowVariableUpperBound}, \eqref{ineq:estimateswitchingVariableUpperBoundSFUJ}, and \eqref{ineq:estimateswitchingVariableUpperBoundUFSJ} at hand, we obtain
\begin{align}
    \Phi_{\sigma(t)}(v(t)) - \Phi_{\sigma(0)}(v_0)  &= \Phi_{\sigma(t)}(v(t)) - \Phi_{\sigma(t_n)}(v_n) + \sum\limits_{i = 1}^n \Phi_{\sigma(t_i)}(v_i)  -  \Phi_{\sigma(t^-_i)}(v^-_i)  +\Phi_{\sigma(t^-_i)}(v^-_i)  -  \Phi_{\sigma(t_{i - 1})}(v_{i - 1}) \nonumber\\
    &\leq \sum_{p \in \mathcal S} N^{\sigma}_p(t_0,t)\tau_p (1 - \delta) -T^\sigma_p(t_0,t)   + \sum_{p \in \mathcal U}  - N^{\sigma}_p(t_0,t)\tau_p (1 + \delta) + T^\sigma_p(t_0,t), \label{ineq:totalEstimateAsymptoticStability1}
\end{align}
where we separated the stable and the unstable modes, and took into account that $\sigma(t_{i-1})=\sigma(t_i^-)$. The first sum can be bounded by the MDADT condition given in Definition~\ref{def:MDADT} as
\begin{align}
    \sum_{p \in \mathcal S}N^{\sigma}_p(t_0,t)\tau_p(1 - \delta) - T^\sigma_p(t_0,t) 
    &= (1 - \delta)\!\para{\sum_{p \in \mathcal S} N^{\sigma}_p(t_0,t) \tau_p - T^\sigma_p(t_0,t) } - \delta \sum_{p \in \mathcal S} T^\sigma_p(t_0,t) \nonumber\\
    &\leq (1 - \delta) T_{\mathcal S} - \delta \sum_{p \in \mathcal S} T^\sigma_p(t_0,t) \label{ineq:estimateMDADT}
\end{align}
for $p \in \mathcal S$; and for $p \in \mathcal U$ the MDALT condition in Definition \ref{def:MDALT} implies
\begin{align}
    \sum_{p \in \mathcal U}  - N^{\sigma}_p(t_0,t)\tau_p (1 + \delta) + T^\sigma_p(t_0,t)
    &= (1 + \delta)\para{\sum_{p \in \mathcal U} - N^{\sigma}_p(t_0,t) \tau_p + T^\sigma_p(t_0,t) } - \delta \sum_{p \in \mathcal U} T^\sigma_p(t_0,t) \nonumber\\
    &\leq (1 + \delta) T_{\mathcal U} - \delta \sum_{p \in \mathcal U} T^\sigma_p(t_0,t) \label{ineq:estimateMDALT}.
\end{align}

From \eqref{ineq:totalEstimateAsymptoticStability1}, \eqref{ineq:estimateMDADT}, and \eqref{ineq:estimateMDALT}, it follows that
\begin{align}
    \Phi_{\sigma(t)}(v(t)) - \Phi_{\sigma(0)}(v_0) 
    & \leq (1 - \delta) T_{\mathcal S} - \delta \!\sum_{p \in \mathcal S} T^\sigma_p(t_0,t) + (1 + \delta) T_{\mathcal U} - \delta \!\sum_{p \in \mathcal U} T^\sigma_p(t_0,t)\nonumber\\
    &= - \delta (t - t_0) + (1 - \delta) T_{\mathcal S} + (1 + \delta) T_{\mathcal U}. \label{ineq:estimateF} 
\end{align}
This means that $\Phi_{\sigma(t)}(v(t)) - \Phi_{\sigma(0)}(v_0)$ is linearly decreasing in $t$ for $t \in [t_0,t')$.

As $\underline \varphi(x) \leq \abs{\varphi_p(x)} \leq \overline \varphi(x)$, it holds that 
\begin{align}\label{innen:estimatePhiLowerAndUpper}
    \underline \Phi(k) - \underline \Phi(l) \geq \Phi_p(k) - \Phi_q(l) \geq \overline \Phi(k) - \overline \Phi(l)
\end{align}
for all $p, q \in \mathcal M$ and all $k \geq l$, $k, l \in \R^+_0$. Then from \eqref{ineq:estimateF}, we obtain 
\begin{align}
    &\underline \Phi(v(t)) - \underline \Phi(v_0) 
    \leq \Phi_{\sigma(t)}(v(t)) - \Phi_{\sigma(0)}(v_0) \leq - \delta (t - t_0) + C \qquad \text{if }v(t) \leq v_0 \label{ineq:estimatePhiLower}\\
    &\overline \Phi(v(t)) - \overline \Phi(v_0) 
    \leq \Phi_{\sigma(t)}(v(t)) - \Phi_{\sigma(0)}(v_0) \leq - \delta (t - t_0) + C \qquad\text{if }v(t) \geq v_0, \label{ineq:estimatePhiUpper}
\end{align}
where $C := (1 - \delta) T_{\mathcal S} + (1 + \delta) T_{\mathcal U}$.

Next, we distinguish between two cases: Case 1: For the case $\inf(\image \underline \Phi) = m > - \infty$, we set $\overline \Phi^{-1}\!\para{s} = 0$ if $s \leq \inf(\image \overline \Phi)$, and define 
\begin{align*}
    \tilde \beta(r,s) := \max\!\left\{\underline \Phi^{-1}\!\!\para{\Gamma\!\para{\underline \Phi\!\para{r}, \delta s}}\!,
    \overline \Phi^{-1}\!\!\para{\overline \Phi\!\para{r} + C - \delta s}\!\right\},
\end{align*}
where
\begin{align*}
    \Gamma(u,v) &:= u + C - (u + C-m)\para{1 - \exp\!\para{\frac{-v}{u + C - m}}}  \geq u + C - v.
\end{align*}

Case 2: If $\inf(\image \underline \Phi) = - \infty$, we define
\begin{align*}
    \tilde \beta(r,s) := \max\!\left\{\underline \Phi^{-1}\!\!\para{\underline \Phi\!\para{r} + C - \delta s}\!,\overline \Phi^{-1}\!\!\para{\overline \Phi\!\para{r} + C - \delta s}\!\right\}\!.
\end{align*}

We made this case distinction because we want $\tilde \beta$ to be strictly falling to zero in the second argument. Note that $\tilde \beta$ is strictly increasing in the first argument.

Thus, from \eqref{ineq:estimatePhiLower} and \eqref{ineq:estimatePhiUpper}, it follows that
\begin{align*}
    v(t) \leq \tilde \beta\!\para{v_0,t  -t_0}
\end{align*}
for all $t \in [t_0,t')$.


Obviously, $\tilde \beta$ is continuous and strictly decreasing to zero for $s \to \infty$ by the definition of $\overline \Phi$ and  $\underline \Phi$. However, it is a priori not clear if $\tilde \beta$ is bounded, and therefore, we cannot conclude that $\tilde \beta \in \KK\LL$. Note that for $t - t_0 \geq \tau := \frac{C}{\delta}$, it follows that  $\tilde \beta(r,t - t_0) \leq r$. It remains to show that we can bound $v(t) = V_\sigma(t,x(t))$ for $t - t_0\leq \tau$.


To this end, from Lemma \ref{lem:LipschitzContinuousFlow},  for $V_{\sigma(t_0)}(t_0,x(t_0)) \geq \chi\!\para{\norm{u}_\infty}$, i.e.,  
\begin{align*}
    \norm{u}_\infty \leq \chi^{-1}(V_{\sigma(t_0)}(t_0,x(t_0))) \leq \chi^{-1}(\alpha_2(\norm{x_0}_X)),
\end{align*}
it follows that
\begin{align*}
    &\max_{t \in [t_0,\tau]}\norm{x(t)}_X \leq L(\norm{x_0}_X,\chi^{-1}(\alpha_2(\norm{x_0}_X)),\tau, \sigma) \norm{x_0}_X.
\end{align*}
Therefore, for $t \in [t_0,t')$ and $t - t_0 \leq \tau$, $\norm{x(t;t_0,x_0,u,\sigma)}$ can be bounded by a $\KK_\infty$-function in $x_0$, uniformly in $t - t_0$. Hence, there exists $\beta \in \mathcal{KL}$ defined by $\beta(r,s) := \alpha_1^{-1}\!\para{\tilde{\beta}\!\para{\alpha_2(r),s}}$ for $t \geq \tau$, such that
\begin{align} \label{ineq:asymptoticStabilityOfTrajectory}
    \norm{x(t;t_0,x_0,u,\sigma)}_X \leq \beta\!\para{\norm{x_0}_X\!, t - t_0}
\end{align}
for $t \in [t_0, t')$ (and by the above argument then also on $[t_0,t_*]$). Note that the constructed $\beta$ is actually independent of $t_0$, i.e.\ the same bound can be used also on a later time interval $[t_1,t_2)$ with $t_2>t_1>t_0$, such that $x(t)\notin A_1(t)$ for all $t\in[t_1,t_2)$.

\textbf{Step 2}: Next, we  show that trajectories that reach to $A_1(t)$ at any time $t \in I$, stay bounded for all times. We define the sets
\begin{align*}
    A_2(t) &:= \braces{x \in X \,\middle| \, V_\sigma(t,x) \leq \gamma_2\!\para{\norm{u}_\infty}}\!, \\
    A_3(t) &:= \braces{x \in X \,\middle| \, V_\sigma(t,x) \leq \gamma_3\!\para{\norm{u}_\infty}},
\end{align*}
where $\gamma_2, \gamma_3 \in \mathcal{K}_\infty$ and are defined by
\begin{align*}
    \gamma_2(s) &= \max\!\braces{\alpha_3(s), \chi(s)}, \\
    \gamma_3(s) &= \max\!\braces{\gamma_2(s), \alpha_2\!\para{\beta\!\para{\alpha_1^{-1}\!\para{\gamma_2(s)}\!,0}}}.
\end{align*}
Obviously, $A_1(t) \subseteq A_2(t) \subseteq A_3(t)$ for all $t \in [t_0,\infty)$, $p \in \mathcal M$. Trajectories leaving $A_1(t)$ by flow have to cross the boundary $\partial A_1(t)$ and trajectories leaving $A_1(t)$ by jump only reach to $A_2(t)$ due to condition \eqref{ineq:LyapunovJump}. In both cases, there exists a time $t' \in [t, \infty)$, such that $x(t') \subseteq A_2(t') \setminus A_1(t'))$. Therefore, we can apply \eqref{ineq:asymptoticStabilityOfTrajectory} combined with \eqref{ineq:LyapunovDefiniteness}, where $t = t_0 = t'$. As a consequence, all the trajectories that leave $A_1(t)$ will stay in $A_3(t)$.

Next, we define $\gamma \in \mathcal{K}_\infty$, $\gamma :=\alpha_1^{-1} \circ \gamma_3$. Then $\norm{x(t;t_0,x_0,u,\sigma)}_X \leq \gamma\!\para{\norm{u}_\infty}$
holds for all $t > t_*$.
From this equation and \eqref{ineq:asymptoticStabilityOfTrajectory}, we can conclude
\begin{align*}
    \norm{x(t;t_0,x_0,u,\sigma)}_X \leq \beta\para{\norm{x_0}_X\!, t - t_0} + \gamma\!\para{\norm{u}_\infty},
\end{align*}
just as desired.
\end{proof}

\begin{remark}\label{rem:linearDwellTimeCond}
    Note that in the case of an impulsive system with only one stable flow and one unstable jump, condition~\eqref{ineq:dwellTimeConditionSFUJ} becomes
    \begin{align*} 
        \int_a^{\psi(a)}\frac{1}{-\varphi(s)}\,\diff s
        \leq \tau(1 - \delta),
    \end{align*}
    where $\varphi_p = \varphi$, $\psi_p = \psi$ as there is only one mode $p \in \mathcal S$. Conversely, for the case of an impulsive system with only one unstable flow and one stable jump, i.e., a single mode $p \in \mathcal U$, condition~\eqref{ineq:dwellTimeConditionUFSJ} reduces to
    \begin{align*} 
        \int_a^{\psi(a)}\frac{1}{-\varphi(s)} \,\diff s
        \geq \tau(1 + \delta).
    \end{align*}
   Thus, the conditions \eqref{ineq:dwellTimeConditionSFUJ} and \eqref{ineq:dwellTimeConditionUFSJ}, for the case of an impulsive system, boils down to the dwell time conditions in \cite{Dashkovskiy2013}.  
   
   Finally, let us discuss the case that the rate functions $\varphi_p(s)$ and $\psi_p(s)$ are linear, i.e.,  $\varphi_p(s) = \eta_p \cdot s$, $\eta_p \in \R\setminus \{0\}$, and $\psi_p(s) = \mu_p \cdot s$, $\mu_p > 0$. Then, it follows that
    \begin{align*}
        \Phi_p(v) 
        = \int_{1}^{v} \frac{1}{\abs{\varphi_p(s)}} \ \diff s 
        = \int_{1}^{v} \frac{1}{\abs{\eta_p}s} \ \diff s 
        = \frac{1}{\abs{\eta_p}} \ln v.
    \end{align*}
    Therefore, in the case $p \in \mathcal S$, \eqref{ineq:dwellTimeConditionSFUJ} reduces to
    \begin{align*}
        \frac{\ln(\tilde \mu_{\sigma(t_i^-)})}{\abs{\eta_{\sigma(t_i^-)}}} 
        &= \frac{1}{\abs{\eta_{\sigma(t_i)}}} \ln(\mu_{\sigma(t_i^-)} \cdot a) - \frac{1}{\abs{\eta_{\sigma(t_i^-)}}} \ln(a) \leq \tau_{\sigma(t_i^-)}(1 - \delta) < \tau_{\sigma(t_i^-)},
    \end{align*}
    where we define $\tilde \mu_{\sigma(t_i^-)} := \mu_{\sigma(t_i^-)}e^{\abs{\eta_{\sigma(t_i^-)}}-\abs{\eta_{\sigma(t_i)}}}$.
    Conversely for $p \in \mathcal U$, from \eqref{ineq:dwellTimeConditionUFSJ}, it follows that
    \begin{align*}
        -\frac{\ln(\tilde \mu_{\sigma(t_i^-)})}{\abs{\eta_{\sigma(t_i^-)}}} 
        &= -\frac{1}{\abs{\eta_{\sigma(t_i)}}} \ln(\mu_{\sigma(t_i^-)} \cdot a) + \frac{1}{\abs{\eta_{\sigma(t_i^-)}}} \ln(a) \geq \tau_{\sigma(t_i^-)}(1 + \delta) > \tau_{\sigma(t_i^-)}.
    \end{align*}
    Thus, the conditions \eqref{ineq:dwellTimeConditionSFUJ} and \eqref{ineq:dwellTimeConditionUFSJ}, for the case of a switched system and non-decreasing ISS-Lyapunov functions with linear rates, boils down to the dwell time conditions in \cite{Yin2023}.  
\end{remark}

\section{Sufficient and necessary condition for ISS}
\label{sec:nece_and_suff}
In this section, we first introduce a more restrictive characterization of ISS-Lyapunov functions, i.e., a decreasing ISS-Lyapunov function as defined below:
\begin{definition}\label{def:decrLF}
    Let $V_\sigma$ be a candidate ISS-Lyapunov function. If for each $p \in \mathcal{M}$, it holds that $\varphi_p \in - \PP$ and $\psi_p \leq \id$, we call $V_\sigma$ a \emph{(decreasing) ISS-Lyapunov function}, which we denote by $W_\sigma$.
\end{definition}
Then, in the following, we prove that the existence of such a time-varying decreasing ISS-Lyapunov function is not only a sufficient but also a necessary condition for ISS of system $\Sigma$. 
\begin{corollary}\label{cor:LFtoISS}
    If there exists an ISS-Lyapunov function for system $\Sigma$ for a given switching signal $\sigma$ as in Definition~\ref{def:decrLF}, 
    then $\Sigma$ is ISS.
\end{corollary}

\begin{proof}
    By the definition of non-decreasing ISS-Lyapunov functions in Definition \ref{def:nonDecrLF}, it holds that $\mathcal S = \mathcal M$. Inequality~\eqref{ineq:dwellTimeConditionSFUJ} is trivially fulfilled because $\psi_{\sigma(t_i^-)}(a) \leq a$ and $\Phi_{\sigma(t_i)}$ is increasing for each $t_i$, and therefore, the left-hand side of \eqref{ineq:dwellTimeConditionSFUJ} is less than or equal to 0. So, we can choose $\tau_p = 0$ for each $p \in \mathcal M$ and as a consequence, \eqref{ineq:MDADT} is always fulfilled. This completes the proof.
\end{proof}

Next, we will provide a converse Lyapunov theorem as follows:
\begin{theorem}\label{thm:converseLyapunovTheorem}
    Let system $\Sigma$ with a given switching signal $\sigma$ be ISS. Then, there exists a decreasing ISS-Lyapunov function for system $\Sigma$.
\end{theorem}
\begin{proof}
    Considering a fixed switching signal $\sigma$, we define
    \begin{align*}
        \widetilde f(t,x,u) &:= f_{\sigma(t)}(t,x,u),\ t \in I \setminus S, \\
        \widetilde g_i(x,u) &:= g_{\sigma(t_i^-)}(t_i^-,x,u),\ t_i \in S, \ i \in \N
    \end{align*}
    for all $(x, u) \in X \times \U$. Note that for every constants $C,D >0$, functions $\widetilde f$ and $\widetilde g$ are locally Lipschitz continuous with respect to $x$ for $u \leq D, t \in I \setminus S$, and locally Lipschitz continuous with respect to $u$  for $x \leq C, i \in \N$. Then, we can treat the impulsive switched system $\Sigma$ as impulsive system
    \begin{align*}
        \widetilde \Sigma\colon \left\{
        \begin{aligned}
            \dot x(t) &= \widetilde f(t,x(t),u(t)), &&t \in I \setminus S, \\
            x(t_i^+) &= \widetilde g_i(x(t_i^-),u(t_i^-)), &&t_i \in S, \ i \in \N.
        \end{aligned}
        \right.
    \end{align*}
    Then, we can employ \cite[Theorem 2]{Bachmann2023characterization} to conclude that ISS of $\Sigma$ implies the existence of a decreasing ISS-Lyapunov function as given in Definition \ref{def:decrLF}.
\end{proof}

\begin{remark}
    By definition, every decreasing ISS-Lyapunov function is also a non-decreasing ISS-Lyapunov function. Since ISS implies the existence of a decreasing ISS-Lyapunov, it follows that ISS implies the existence of a non-decreasing ISS-Lyapunov.
\end{remark}

\section{Construction of ISS-Lyapunov functions}
\label{sec:construction}
Now, let us assume that the non-decreasing ISS-Lyapunov function as given in Definition \ref{def:nonDecrLF} is in fact decreasing along trajectories. We will show that, in contrast to Theorem~\ref{thm:ISS}, it is not necessary to impose the dwell/leave time conditions \eqref{ineq:MDADT} and \eqref{ineq:MDALT} in order to conclude ISS from the existence of an ISS-Lyapunov function. That is why we will show, in the following, how to construct a (decreasing) ISS-Lyapunov function from a non-decreasing one, provided the switching signal satisfies the corresponding dwell/leave time conditions of Theorem~\ref{thm:ISS}. This construction is motivated by the facts that (i) the existence of such a decreasing ISS-Lyapunov function is a necessary and sufficient condition for ISS, (ii) the level sets of decreasing ISS-Lyapunov functions directly indicate the reachable sets, and (iii) it facilitates the computation of $\beta$ which, in turn, determines the convergence rate.



\begin{theorem}\label{thm:ConstructionOfLF}
    Let $V_\sigma$ be a non-decreasing ISS-Lyapunov functions for system $\Sigma$ with $\varphi_p, \psi_p$ and $\tau_p$ for $p \in \mathcal M$ defined as in Definition \ref{def:nonDecrLF}. Let $\image(\overline \Phi) = \R$.

    Then, an ISS-Lyapunov function for system $\Sigma$ is given by
    \begin{align*}
        W_\sigma(t,x) &= \Phi_{\sigma(t^-)}^{-1}\!\left( \Phi_{\sigma(t)} \!\para{V_\sigma(t,x)} + h_\sigma(\sigma(t),t)\right),
    \end{align*}
    where
    \begin{align*}
        h_\sigma (p, t) 
        = \min_{j \in \{0, \dots, i\}}\!\left\{0,\para{\sum_{p \in \mathcal S} T^\sigma_p(t_j,t) - \tau_p N^{\sigma}_p(t_j^-,t)}(1 - \delta)- \para{\sum_{p \in \mathcal U} T^\sigma_p(t_j,t) - \tau_p N^{\sigma}_p(t_j,t)}(1 + \delta)\right\}
    \end{align*}
    for $t \in [t_i,t_{i + 1})$, $i \in \N_0$.    
\end{theorem}

Before proving Theorem~\ref{thm:ConstructionOfLF}, we discuss the intuition behind the construction of the decreasing Lyapunov function from a non-decreasing one. To this end, consider a system with uniform fixed switching instants $(t_i)_{i \in \N}$ as sketched in Figure~\ref{fig:decrLF}. Note that for a stable mode $p \in \mathcal S$, the Lyapunov value of the non-decreasing ISS-Lyapunov function $V_\sigma(t,x(t))$ is falling faster than the successive (unstable) jump. Therefore, when the trajectory is a priori known, one can connect the maximum point at the beginning of the stable interval with the point directly after the jump to obtain a strictly decreasing ISS-Lyapunov function $W_\sigma(t,x(t))$. As the trajectory is, in general, not a priori known, the next jump height has to be estimated, which is implicitly done by the MDADT Condition~\eqref{ineq:MDADT} and Condition~\eqref{ineq:dwellTimeConditionSFUJ}, which $V_\sigma$ has to satisfy. 
\\
For an unstable mode $p \in \mathcal U$, the scenario is analogous. The flows might not be falling and even increasing but in general, the jumps are stabilizing the overall behavior. Therefore, when the trajectory is a priori known, one can connect the point at the beginning of the unstable interval with the point directly after the jump to obtain a strictly decreasing ISS-Lyapunov function $W_\sigma(t,x(t))$. As the trajectory is not a priori known, this effect is implicitly determined by the MDALT Condition~\eqref{ineq:MDALT} and Condition \eqref{ineq:dwellTimeConditionUFSJ}, which $V_\sigma$ has to satisfy.
\\
Now, to consider a more general system which satisfies an average dwell-time condition: here, additionally to estimating the jump heights, we also have to estimate the switching times. For this, a correction term $h$ as depicted in Figure \ref{fig:correctionFunction} has to be introduced, which measures if the next jump is already overdue or already too many jumps have occurred as compared to the $\tau_p$ in the MDADT and MDALT conditions \eqref{ineq:MDADT} and \eqref{ineq:MDALT}, respectively. These terms appear piecewise linearly in $W_\sigma$ because $\Phi_\sigma$ maps the Lyapunov values to sub-linear functions.

\begin{figure}[htb]
    \begin{tikzpicture}
        \begin{axis}[
            axis x line=center,
            axis y line=center,
            xlabel style={below right},
            ylabel style={above},
            minor xtick={1,2,3},
            grid=minor,
            xmin=0,
            xmax=3.1,
            ymin=-0.35,
            ymax=1.35,
            xlabel=time $t$,
            ylabel={$V_\sigma$, $W_\sigma$},
            xtick={1,2,3},
            xticklabels={$t_1$,$t_2$,$t_3$},
            xticklabel style={font=\tiny},
            extra x ticks ={0.5,1.5,2.5},
            extra x tick labels={${\sigma(t) \in \mathcal S}$, $\sigma(t) \in \mathcal U$, $\sigma(t) \in \mathcal S$},
            extra x tick style={
                every tick/.append style={draw=none},
                xticklabel style = {font=\normalsize,yshift=-.3cm}
            },
            ytick=\empty,
            width=.95\columnwidth,
            height=.60\columnwidth,
            set layers=tick labels on top
            ]
            \draw[draw=none,fill=white] (0.1,-0.07) rectangle (3.2,-0.18);
            
            \addplot[domain=0:1,color=customblue,thick,smooth] {exp(-1.2*x)};
            \addplot[domain=1:2,color=customblue,thick,smooth,forget plot] {exp(+0.2*x-0.7)};
            \addplot[domain=2:3,color=customblue,thick,smooth,forget plot] {exp(-2*x+3)};
            \addplot[domain=3:3.25,color=customblue,thick,smooth,forget plot] {exp(+0.2*x-2.1)};
            \addplot[mark=*,only marks,color=customblue,forget plot] coordinates {(1,{exp(-0.5)})(2,{exp(-1)})(3,{exp(-1.5)})};
            \addplot[mark=*,fill=white,only marks,draw=customblue,thick,forget plot] coordinates {(1,{exp(-1.2)})(2,{exp(-0.3)})(3,{exp(-3)})};
            \draw[dotted,color=customblue,thick,smooth] (1,{exp(-0.5)})--(1,{exp(-1.2)});
            \draw[dotted,color=customblue,thick,smooth] (2,{exp(-1)})--(2,{exp(-0.3)});
            \draw[dotted,color=customblue,thick,smooth] (3,{exp(-1.5)})--(3,{exp(-3)});
            
            \addplot[color=red,thick,smooth] {exp(-0.5*x)};
            
            \legend{${V_\sigma(t,x(t))}$, ${W_\sigma(t,x(t))}$};
        \end{axis}
    \end{tikzpicture}
	\caption{Construction of a decreasing ISS-Lyapunov function $W_\sigma(t,x(t))$ from a non-decreasing ISS-Lyapunov function $V_\sigma(t,x(t))$ and their relation.\label{fig:decrLF}}
\end{figure}
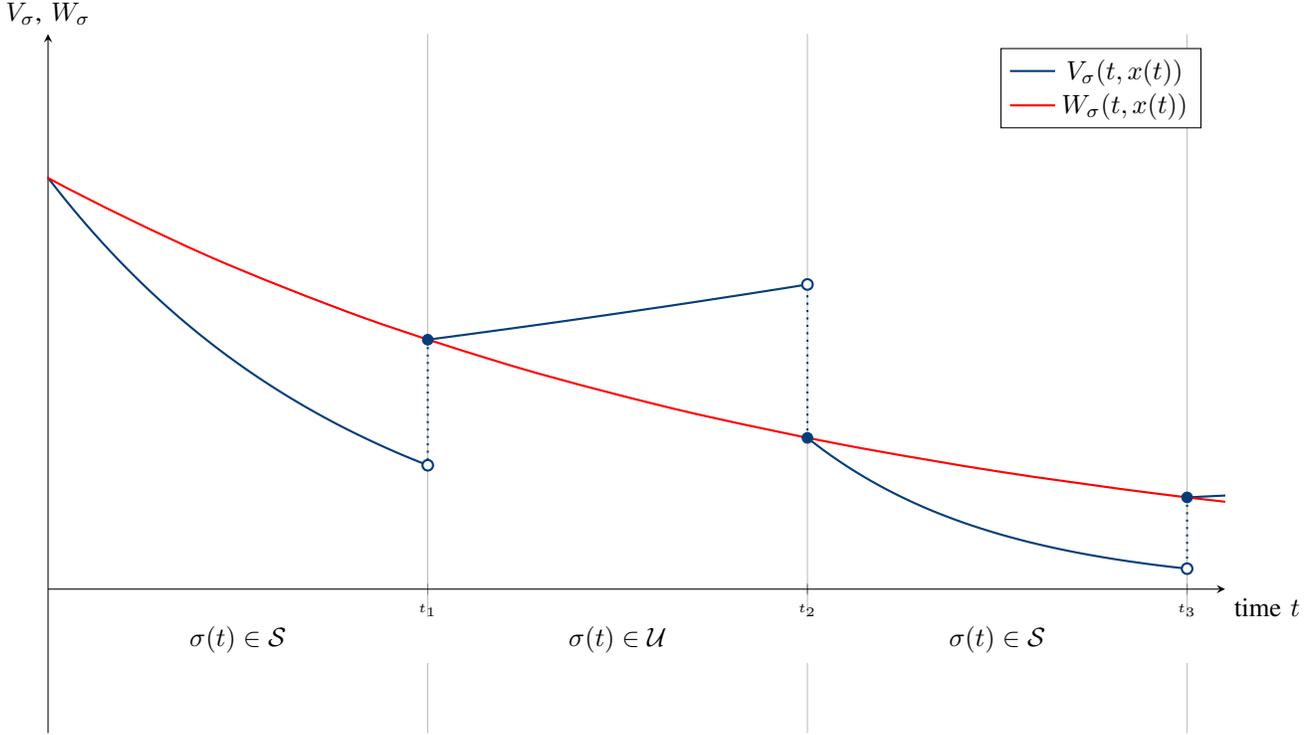

\begin{figure}[htb]
    \begin{tikzpicture}
        \begin{axis}[
            axis x line=center,
            axis y line=center,
            xlabel style={below right},
            ylabel style={above},
            xmin=0,
            xmax=2.38,
            ymin=-0.25,
            ymax=1.35,
            minor xtick={1,2,3},
            minor ytick={1},
            grid=minor,
            xlabel=time $t$,
            ylabel={$h_\sigma(\sigma(t),t)$},
            xtick={0.2,0.3,0.9,1,1.2,1.3,1.9,2,2.3,2.4,2.5,3},
            xticklabels={$t_1$,$t_2$,$t_3$,$t_4$,$t_5$,$t_6$,$t_7$,$t_8$,$t_9$,$t_{10}$,$t_{11}$,$t_{12}$},
            xticklabel style={font=\tiny},
            extra x ticks ={0.5,1.5,2.5},
            extra x tick labels={${\sigma(t) \in \mathcal S}$, $\sigma(t) \in \mathcal U$, $\sigma(t) \in \mathcal S$},
            extra x tick style={
                every tick/.append style={draw=none},
                xticklabel style = {font=\normalsize,yshift=-.3cm}
            },
            ytick={0,1},
			yticklabels=\empty,
            ymajorgrids=true,
            extra y ticks={0,1},
            extra y tick labels={{\tiny${-T_{\mathcal S}(1 - \delta)\quad}$ ${-T_{\mathcal U}(1 + \delta)}$},0},
            yticklabel style={text width=1.4cm,align=right},
            width=.85\columnwidth,
            height=.75\columnwidth,
            set layers=tick labels on top
            ]
            \draw[draw=none,fill=white] (0.1,-0.05) rectangle (3.2,-0.14);

            \draw[color=customblue,thick,smooth] (0,1)--(0.2,1)(0.2,0.8)--(0.3,0.9)(0.3,0.7)--(0.6,1)--(0.9,1)(0.9,0.8)--(1,0.9);
            \addplot[mark=*,only marks,color=customblue,forget plot] coordinates {(0.2,0.8)(0.3,0.7)(0.9,0.8)};
            \addplot[mark=*,fill=white,only marks,draw=customblue,thick,forget plot] coordinates {(0.2,1)(0.3,0.9)(0.9,1)(1,0.9)};
            \draw[dotted,color=customblue,thick,smooth] (0.2,1)--(0.2,0.8)(0.3,0.9)--(0.3,0.7)(0.9,1)--(0.9,0.8);
            \draw[color=customblue,thick,smooth] (1,0.7)--(1.2,0.5)(1.2,0.9)--(1.3,0.8)(1.3,1)--(1.9,0.4)(1.9,0.8)--(2,0.7);
            \addplot[mark=*,only marks,color=customblue,forget plot] coordinates {(1,0.7)(1.2,0.9)(1.3,1)(1.9,0.8)};
            \addplot[mark=*,fill=white,only marks,draw=customblue,thick,forget plot] coordinates {(1.2,0.5)(1.3,0.8)(1.9,0.4)(2,0.7)};
            \draw[dotted,color=customblue,thick,smooth] (1,0.9)--(1,0.7)(1.2,0.5)--(1.2,0.9)(1.3,0.8)--(1.3,1)(1.9,0.4)--(1.9,0.8);
            \draw[color=customblue,thick,smooth] (2,1)--(2.3,1)(2.3,0.8)--(2.4,0.9)(2.4,0.7)--(2.5,0.8)(2.5,0.6)--(2.9,1)--(3,1);
            \addplot[mark=*,only marks,color=customblue,forget plot] coordinates {(2,1)(2.3,0.8)(2.4,0.7)(2.5,0.6)};
            \addplot[mark=*,fill=white,only marks,draw=customblue,thick,forget plot] coordinates {(2.3,1)(2.4,0.9)(2.5,0.8)};
            \draw[dotted,color=customblue,thick,smooth] (2,0.7)--(2,1)(2.3,1)--(2.3,0.8)(2.4,0.9)--(2.4,0.7)(2.5,0.8)--(2.5,0.6);
            \draw[color=customblue,thick,smooth] (3,0.8)--(3.3,0.5);
            \addplot[mark=*,only marks,color=customblue,forget plot] coordinates {(3,0.8)};
            \addplot[mark=*,fill=white,only marks,draw=customblue,thick,forget plot] coordinates {(3,1)};
            \draw[dotted,color=customblue,thick,smooth] (3,1)--(3,0.8);

            \draw [decorate, decoration = {brace,raise=3pt}] (0.9,0.8)--(0.9,1);
            \node at (0.68,0.85) {\footnotesize$\tau_{\sigma(t_3^-)}$};
            
            \draw [decorate, decoration = {brace,raise=3pt}] (1.9,0.8)--(1.9,0.4);
            \node at (2.17,0.57) {\footnotesize$\tau_{\sigma(t_7^-)}$};
     
        \end{axis}
    \end{tikzpicture}        
	\caption{Illustration of function $h_\sigma(p,t)$ for a fixed switching signal.\label{fig:correctionFunction}}
\end{figure}

\begin{proofof}[Proof of Theorem~\ref{thm:ConstructionOfLF}]
    Let us verify the conditions in Definition \ref{def:candLyapunovFunction} for $W_\sigma$ to be an ISS-Lyapunov function. Before proceeding, note that, by \eqref{ineq:MDADT} and \eqref{ineq:MDALT}, it follows for all $p \in \mathcal M$, $t \in I$ and all switching signals $\sigma \in \Omega$ that
\begin{align*}
    h_\sigma(p,t) \in [-  T_{\mathcal S}(1 - \delta) - T_{\mathcal U}(1 + \delta),0].
\end{align*}

\emph{Condition 1}: It follows for $(t,x, u) \in (I\setminus S) \times X \times U$ that
\begin{align*}
    &W_\sigma(t,x) \\
    &\geq \Phi_{\sigma(t^-)}^{-1}\!\left( \Phi_{\sigma(t)} \!\para{V_\sigma(t,x)} -  T_{\mathcal S}(1 - \delta) - T_{\mathcal U}(1 + \delta)\right) \\
    &\geq \Phi_{\sigma(t^-)}^{-1}\!\left( \Phi_{\sigma(t)} \!\para{\alpha_1\!\para{\norm{x}_X}} -  T_{\mathcal S}(1 - \delta) - T_{\mathcal U}(1 + \delta)\right)\\
    &\geq
    \min\!\braces{\underline \Phi^{-1},\overline \Phi^{-1}}\!\left( \min\!\braces{\overline \Phi, \underline \Phi}\!\para{\alpha_1\!\para{\norm{x}_X}} -  T_{\mathcal S}(1 - \delta) - T_{\mathcal U}(1 + \delta)\right)\\
    &=: \widetilde\alpha_1\!\para{\norm{x}_X}.
\end{align*}
From \eqref{innen:estimatePhiLowerAndUpper} together with $\image(\overline \Phi) = \R$ it follows that $\widetilde \alpha_1$ is a $\KK_\infty$-function. 
Furthermore, we have
\begin{align*}
    W_\sigma(t,x) 
    &\leq \Phi_{\sigma(t^-)}^{-1}\!\para{\Phi_{\sigma(t)}\!\para{V_\sigma(t,x)} + 0}  \\
    &=V_\sigma(t,x)
    \leq \alpha_2\!\para{\norm{x}_X}
    \\
    &=: \widetilde\alpha_2\!\para{\norm{x}_X}.
\end{align*}
Hence, Condition 1 is shown if $t$ is not a switching instant. 
 
 If $t$ is a switching instant,  then we have to consider two cases. If $\sigma(t^-) \in \mathcal S$, then we have
 \begin{align}
    W_\sigma(t_i,g_{\sigma(t_i^-)}(t_i^-,x,u)) 
    &= \Phi_{\sigma(t_i^-)}^{-1} \!\!\para{\Phi_{\sigma(t_i)} \!\para{V_\sigma\!\para{t_i, g_{\sigma(t_i^-)}\!\para{t_i^-,x, u}}} + h_\sigma(\sigma(t_i),t_i)} \nonumber\\
    &\leq \Phi_{\sigma(t_i^-)}^{-1}\!\left(\Phi_{\sigma(t_i)} \!\para{\psi_{\sigma(t_i^-)}\!\para{V_\sigma(t_i^-,x)}} + h_\sigma(\sigma(t_i^-),t_i^-) - \tau_{\sigma(t_i^-)}(1 - \delta)\right) \nonumber\\
    &\leq \Phi_{\sigma(t_i^-)}^{-1}\!\para{\Phi_{\sigma(t_i^-)}\!\para{V_\sigma(t_i^-,x)} + h_{\sigma}(\sigma(t_i^-),t_i^-) } \nonumber\\
    &= W_\sigma(t_i^-,x), \label{ineq:LyapunovJumpEstimateSFUJ}
\end{align}
where we used the definition of $h_\sigma$ in the second step and the inequality \eqref{ineq:dwellTimeConditionSFUJ} in the third step.

Now, if $\sigma(t^-) \in \mathcal U$, then we have
\begin{align}
    W_\sigma(t_i,g_{\sigma(t_i^-)}(t_i^-,x,u))
    &= \Phi_{\sigma(t_i^-)}^{-1} \!\!\para{\Phi_{\sigma(t_i)} \!\para{V_\sigma\!\para{t_i, g_{\sigma(t_i^-)}\!\para{t_i^-,x, u}}} + h_\sigma(\sigma(t_i),t_i)} \nonumber\\
    &\leq \Phi_{\sigma(t_i^-)}^{-1}\!\left(\Phi_{\sigma(t_i)} \!\para{\psi_{\sigma(t_i^-)}\!\para{V_\sigma (t_i^-,x)}} + h_\sigma(\sigma(t_i^-),t_i^-) + \tau_{\sigma(t_i^-)}(1 + \delta)\right) \nonumber\\
    &\leq \Phi_{\sigma(t_i^-)}^{-1}\!\para{\Phi_{\sigma(t_i^-)}\!\para{V_\sigma(t_i^-,x)} + h_\sigma(\sigma(t_i^-),t_i^-) } \nonumber\\
    &= W_\sigma(t_i^-,x), \label{ineq:LyapunovJumpEstimateUFSJ}
\end{align}
where we used the definition of $h_\sigma$ in the second step and the inequality \eqref{ineq:dwellTimeConditionUFSJ} in the third step.

\emph{Condition 2:} Let $\widetilde \chi = \chi$ and $W_\sigma(t,x) \geq \widetilde \chi\!\para{\norm{u}_\infty}$. Then it holds that $V_\sigma(t,x) \geq W_\sigma(t,x) \geq \widetilde \chi\!\para{\norm{u}_\infty}$.

\emph{Part a:} We bound the flow behavior for $t\in I\setminus S$.\\
Let us first consider the case $\sigma(t) \in \mathcal S$, but $h_\sigma(\sigma(t),t) \neq 0$. It follows
\begin{align*}
    \dot W_\sigma(t,x)
    &\leq \para{\Phi_{\sigma(t^-)}^{-1}}'\!\para{\Phi_{\sigma(t)}\!\para{V_\sigma(t,x)} + h_\sigma(\sigma(t),t)} \para{\tfrac{\dot V_\sigma(t,x)}{-\varphi_{\sigma(t)}\!\para{V_\sigma(t,x)}} + 1\cdot (1 - \delta)} \nonumber\\
    &\leq - \delta |\varphi_{\sigma(t)}|\!\para{\Phi_{\sigma(t)}^{-1}\!\para{\Phi_{\sigma(t)}\!\para{V_\sigma(t,x)} + h_\sigma(\sigma(t),t)}} \nonumber\\
    &= - \delta |\varphi_{\sigma(t)}|\!\para{W_\sigma(t,x)}, 
\end{align*}
where $\para{\Phi_{\sigma(t^-)}^{-1}}'$ denotes the derivative of $\Phi_{\sigma(t^-)}^{-1} = \Phi_{\sigma(t^-)}^{-1}(s)$ with respect to $s \in \image (\Phi_{\sigma(t^-)})$.
Here, we used the inverse function theorem, \eqref{ineq:transformIDEforV_flow}, and $\sigma(t^-) = \sigma(t)$ in the second step. Let us now consider the case  $\sigma(t) \in \mathcal S$ and $h_\sigma(\sigma(t),t) = 0$. Then it holds that
\begin{align*}
    \dot W_\sigma(t,x)
    &\leq \para{\Phi_{\sigma(t^-)}^{-1}}'\!\para{\Phi_{\sigma(t)}\!\para{V_\sigma(t,x)} + h_\sigma(\sigma(t),t)}  \para{\tfrac{\dot V_\sigma(t,x)}{-\varphi_{\sigma(t)}\!\para{V_\sigma(t,x)}}} \nonumber\\
    &\leq - \abs{\varphi_{\sigma(t)}}\!\para{\Phi_{\sigma(t)}^{-1}\!\para{\Phi_{\sigma(t)}\!\para{V_\sigma(t,x)} + h_\sigma(\sigma(t),t)}} \nonumber\\
    &= - \abs{\varphi_{\sigma(t)}}\!\para{W_\sigma(t,x)}.
\end{align*}

In the case $\sigma(t) \in \mathcal U$, we obtain
\begin{align*}
    \dot W_\sigma(t,x)
    &\leq \para{\Phi_{\sigma(t^-)}^{-1}}'\!\para{\Phi_{\sigma(t)}\!\para{V_\sigma(t,x)} + h_\sigma(\sigma(t),t)}  \para{\tfrac{\dot V_\sigma(t,x)}{\varphi_{\sigma(t)}\!\para{V_\sigma(t,x)}} - 1 \cdot (1 + \delta)} \nonumber\\
    &\leq - \delta \cdot \varphi_{\sigma(t)}\!\para{\Phi_{\sigma(t)}^{-1}\!\para{\Phi_{\sigma(t)}\!\para{V_\sigma(t,x)} + h_\sigma(\sigma(t),t)}} \nonumber\\
    &= - \delta \ \varphi_{\sigma(t)}\!\para{W_\sigma(t,x)}.
\end{align*}
Therefore, for $\widetilde \varphi_p := \min\{\delta,1\} \abs{\varphi_p}$ for $p \in M$, the inequality $\dot W_\sigma(t,x) \leq - \widetilde \varphi_{\sigma(t)} \!\para{W_\sigma(t,x)}$ holds for all $t \in I \setminus S$.

\emph{Part b:} We verify the jump behavior for $t=t_i\in S$.\\
If $\sigma(t_i^-) \in \mathcal S$, then \eqref{ineq:LyapunovJumpEstimateSFUJ} implies $$W_\sigma(t_i,g_{\sigma(t_i^-)}(t_i^-,x,u(t_i^-))) \leq W_\sigma(t_i^-,x).$$ On the other hand, if $\sigma(t_i^-) \in \mathcal U$, then \eqref{ineq:LyapunovJumpEstimateUFSJ} implies $W_\sigma(t_i,g_{\sigma(t_i^-)}(t_i^-,x,u(t_i^-))) \leq W_\sigma(t_i^-,x)$, which is also valid for the case $h_\sigma(\sigma(t_i),t_i)= 0$, i.e., $h_\sigma(\sigma(t_i),t_i)  \leq h_\sigma(\sigma(t_i^-),t_i^-) + \tau_{\sigma(t_i^-)}(1+ \delta)$. Hence,
\begin{align}\label{ineq:LyapunovJumpEstimate}
    W_\sigma(t_i,g_{\sigma(t_i^-)}(t_i^-,x,u(t_i^-))) \leq W_\sigma(t_i^-,x)
\end{align}
holds true for all $t_i \in S$.

\emph{Condition 3}: For $W_\sigma(t_i^-,x) < \widetilde \chi\!\para{\norm{u}_\infty}$, either $V_\sigma(t_i^-, x) < \chi\!\para{\norm{u}_\infty}$ or $V_\sigma(t_i^-, x) \geq \chi\!\para{\norm{u}_\infty}$ holds. In the former case, we have
\begin{align*}
    W_\sigma(t_i,g_{\sigma(t_i^-)}(t_i^-,x,u(t_i^-)))
    &\leq V_\sigma(t_i^-, g_{\sigma(t_i^-)}(t_i^-,x,u(t_i^-)))\\
    &\leq \alpha_3\!\para{\norm{u}_\infty},
\end{align*}
while in the latter case the estimate \eqref{ineq:LyapunovJumpEstimate} implies
\begin{align*}
    W_\sigma(t_i,g_{\sigma(t_i^-)}(t_i^-,x,u(t_i^-)))
    &\leq W(t_i^-,x) \leq \widetilde\chi\!\para{\norm{u}_\infty}.
\end{align*}
Hence, we can choose $\widetilde \alpha_3 \in \mathcal K$ as $$\widetilde \alpha_3(a) := \max\{\alpha_3(a),\widetilde \chi(a) \}.$$
This concludes the proof.
\end{proofof}

\section{Robustness with respect to the switching signal}
\label{sec:robustness}
The ISS result given in Theorem \ref{thm:ISS} applies to system $\Sigma$ with a predefined switching signal via a non-decreasing ISS-Lyapunov function tailored to that particular switching signal.
However, in practice, the switching signal is often a priori unknown, uncertain, or subject to perturbations. This motivates us to provide an ISS result for a set of switching signals in which only the possible mode changes are prescribed but the exact switching instants and the order in which the modes appear are unknown.
To this end, we define $\mathcal{Q}\subseteq \mathcal{M}\times\mathcal{M}$ as the set of pairs $(p,q)$ corresponding to the allowed mode changes, i.e.,\ $p=\sigma(t_i)$ and $q=\sigma(t_i^-)$ for any $\sigma \in \Omega$ and a set of switching times $(t_i)_{i \in \N}$, and consider the following class of impulsive switched systems with time-invariant flow and jumps:
\begin{align*}
    \widehat\Sigma\colon \left\{
    \begin{aligned}
        \dot{x}(t) &= f_{\sigma(t)}\!\para{x(t),u(t)}, \quad &&t \in I\setminus S, \\
        x(t_i^+) &=g_{\sigma(t_i^-)}\!\para{x(t_i^-),u(t_i^-)}, \quad &&t_i\in S.
    \end{aligned}
    \right.
\end{align*}
Note that to achieve such a robust ISS result with respect to sets of switching signals corresponding to $Q$, we cannot use the non-decreasing ISS-Lyapunov functions $V_\sigma$ given in Definition~\ref{def:nonDecrLF}. This is because the non-decreasing ISS-Lyapunov function $V_\sigma$ is composed of $(\widetilde V_p)_{p \in M}$, which only exists on the intervals $I_p^\sigma$ on which the $p$-th mode is active. Therefore, now we allow $(\widetilde V_p)_{p \in M}$ in Definition~\ref{def:nonDecrLF} to be time-invariant and evolve beyond the active switching interval $I_p^\sigma$. 
From this, we will conclude ISS for $\Sigma$, which is robust with respect to the switching signal.

We can now formulate the following corollary to Theorem~\ref{thm:ISS}, which gives robustness with respect to the set of switching signals in $\Omega$.

\begin{corollary}\label{cor:robustISS}
    Let $\mathcal Q$ be a set of mode changes, and consider the impulsive switched system $\widehat\Sigma$. Let $\mathcal{M} = \mathcal{S} \,\dot{\cup}\, \mathcal{U}$. Let $\widetilde V_p \in \Cont(X, \R^+_0)$ for $p \in M$, such that all of the following conditions hold true:
    \begin{enumerate}
		\item There exist functions $\alpha_1, \alpha_2 \in \mathcal{K}_\infty$, such that
		\begin{align*}
			\alpha_1\!\para{\norm{x}_X} \leq \widetilde V_p(x) \leq \alpha_2\!\para{\norm{x}_X} 
		\end{align*}
		holds true for all $t \in I$, all $x \in X$ and  for each $p\in \mathcal M$. 
		\item There exist functions $\chi \in \mathcal{K}_\infty$ and $\underline \varphi, \overline \varphi \in \PP$, such that for each $(p,q) \in \mathcal Q$, there exist $\psi_p \in \PP$ and $\varphi_p$, such that $-\varphi_p \in \mathcal P$ for $p \in \mathcal S$, and $\varphi_p\in\mathcal{P}$ for all $p\in\mathcal{U}$ with $\underline \varphi(s) \leq\abs{\varphi_p(s)} \leq \overline \varphi(s)$ $\forall s \in \R^+_0$, such that for all inputs $u \in U$ and all $x \in X$, the inequality
		\begin{align*}
			&\para{\tfrac{\diff}{\diff x} \widetilde V_p(x)}^T\cdot f(x,u) \leq \varphi_p\!\para{\widetilde V_p(x)} 
        \end{align*}
        holds true, whenever $\widetilde V_p(x) \geq \chi\!\para{\norm{u}_\infty}$, and
        \begin{align*}
			&\widetilde V_p(x^+) \leq \psi_q\!\para{\widetilde V_p(x)} 
		\end{align*}
		holds true, whenever $\widetilde V_q(x) \geq \chi\!\para{\norm{u}_\infty}$, where $x^+=g_{q}(x, u)$.
		\item There exists a function $\alpha_3 \in \mathcal{K}$, such that for all $x \in X$, all $u \in U_c$, which satisfy $\widetilde V_q(x) < \chi\!\para{\norm{u}_\infty}$, the jump rule satisfies
		\begin{align*}
	    	\widetilde V_p(x^+) &\leq \alpha_3\!\para{\norm{u}_\infty}. 
		\end{align*}
	\end{enumerate}
    Let $\Omega$ be the set of switching signals for which all mode changes are corresponding to the set of mode changes $\mathcal Q$.
    
    Let every switching signal $\sigma \in \Omega$ satisfy
    the MDADT condition \eqref{ineq:MDADT} for all modes in $\mathcal{S}$ and the MDALT condition \eqref{ineq:MDALT} for all modes in $\mathcal{U}$ with corresponding dwell/leave times $\{\tau_p\}_{p \in \mathcal{S}}$ and $\{\tau_p\}_{p \in \mathcal{U}}$. Let there exist a $\delta > 0$ such that for all $a > 0$, 
    one of the following two inequalities is satisfied:
    \begin{align*} 
        \Phi_p(\psi_q(a)) - \Phi_q(a) \leq \tau_q(1 - \delta),
    \end{align*}
    if $q\in\mathcal{S}$, or
    \begin{align*} 
        -\Phi_p(\psi_q(a)) + \Phi_q(a) \geq \tau_q(1 + \delta),
    \end{align*}
    if $q \in\mathcal{U}$.

    Then, $\widehat\Sigma$ is ISS, and the function
    $V: I \times X \to \R_0^+:$ $V_\sigma(t,x) := \widetilde V_{\sigma(t)}(x)$ for all $(t,x) \in I \times X$ defines a candidate ISS-Lyapunov function.
\end{corollary}
\begin{proof}
    For fixed $\sigma$, $V_\sigma$ is a non-decreasing ISS-Lyapunov function. The claim then follows from Theorem \ref{thm:ISS}.
\end{proof}

\begin{definition}\label{def:famcandLyapunovFunction}
    Let $\Omega$ be a set of switching signals, and $\widehat \Sigma$ be a switched system. 
    Let $V: \Omega \times I \times X \to \R_0^+$ be a function, such that for each switching signal $\sigma \in \Omega$, we have that $V_\sigma(t,x) := V(\sigma,t,x) = \widetilde V_{\sigma(t)}(x)$ for all $(\sigma,t,x) \in \Omega \times I \times X$.
    Then, we call $\{V_\sigma\}_{\sigma \in \Omega}$ a \emph{family of candidate ISS-Lyapunov functions}.
\end{definition}

\begin{remark}
    Note that, if constant for forward-completeness $K = K(C, D, \tau, \sigma)$ and $\delta$ are independent of $\sigma$, then $\beta$ and $\gamma$ can be chosen independently of $\sigma$, i.e., $\Sigma$ is ISS uniformly with respect to $\sigma \in \Omega$.
\end{remark}

Now, we apply Corollary \ref{cor:robustISS} to provide a sufficient condition for ISS of linear impulsive switched systems in terms of LMIs. We achieve this via candidate ISS-Lyapunov function given in Definition \ref{def:famcandLyapunovFunction} for the case when it is quadratic, i.e., $V_p(x) := x^TM_px$ for some $M_p \in \R^{n \times n}$ for $x \in X$, $p \in \mathcal M$. Furthermore, for ease of illustration, we only consider systems with finitely many modes $p$.
\begin{theorem}
 Consider a linear impulsive switched system $\widehat\Sigma$:
    \begin{align*}
        \dot x(t) &= A_{\sigma(t)} x(t) + B_{\sigma(t)} u(t),\ t \in I\setminus S,\\
        x(t_i) &= J_{\sigma(t_i^-)} x(t_i^-) + H_{\sigma(t_i^-)} u(t_i^-), \ t_i \in S,
    \end{align*}
    where $A_p \in \R^{n \times n}$, $B_p \in \R^{n \times m}, J_p \in \R^{n \times n}$ and $H_p \in \R^{n \times m}$ for $p \in \mathcal M$ and let $\sigma$ be subject to the mode changes $\mathcal{Q}$.

    If there exist symmetric positive definite matrices $M_p \in \R^{n \times n}, Q_p \in \R^{m \times m}$, $p \in \mathcal M$, and constants $\eta_p \in \R$ and $\mu_p > 0$, such that the LMIs
    \begin{align}
        \begin{pmatrix}
            A_p^T M_p + M_p A_p - \eta_p M_p &M_pB_p \\
            B_p^TM_p & -Q_p
        \end{pmatrix}
        \leq 0, \label{ineq:LMI1}\\
        \begin{pmatrix}
            J_q^TM_p J_q - \mu_q M_q &J_q^TM_pH_q \\
            H_q^TM_p J_q & H_q^T M_p H_q - Q_q
        \end{pmatrix}
        \leq 0 \label{ineq:LMI2}
    \end{align}
    are satisfied for all $(p,q)\in\mathcal{Q}$, then $\widehat\Sigma$ is ISS for each switching sequence $\sigma$ satisfying 
    the MDADT condition \eqref{ineq:MDADT} for all modes in $\mathcal{S}$ and the MDALT condition \eqref{ineq:MDALT} for all modes in $\mathcal{U}$ with corresponding dwell/leave times $\{\tau_p\}_{p \in \mathcal{S}}$ and $\{\tau_p\}_{p \in \mathcal{U}}$ with
    \begin{align}
        \eta_q &< 0 \land \frac{\ln \mu_q}{\abs{\eta_p}} \leq \tau_q (1 - \delta), \ \forall q \in \mathcal S. \label{ineq:LMI3}\\
        \eta_q &\geq 0 \land -\frac{\ln \mu_q}{\abs{\eta_p}} \geq \tau_q (1 + \delta), \ \forall q \in \mathcal U. \label{ineq:LMI4}
    \end{align}
    
    Moreover, a non-decreasing ISS-Lyapunov function for $\widehat\Sigma$ is given by by $V_\sigma(x) = x^T M_{\sigma(t)} x$.
\end{theorem}

\begin{proof}
    Let \eqref{ineq:LMI3} and \eqref{ineq:LMI4} be satisfied. We show that $V_p(x) = x^T M_p x$ defines a nondecreasing ISS-Lyapunov function in dissipation form. Let $\chi(s) := \lambda s^2$ for $s \in \R^+_0$, where $\lambda$ is the maximum eigenvalue of $Q_p$ for all $p \in  \mathcal M$.
    Then, the time derivative of $V_p$ along the solutions of the continuous dynamics of $\widehat\Sigma$ satisfies
    \begin{align*}
        \dot V_p(x) &= \dot x^T M_p x + x^T M_p \dot x \\
        &= (A_p x + B_p u)^T M_p x + x^T M_p (A_p x + B_p u)  \\
        &= x^T A_p^T M_p x + x^T M_p A_p x + u^T B_p^T M_p x + x^T M_p B_p u \\
        &\leq \eta_p x^T M_p x + u^T Q_p u \\
        &\leq \eta_p  V_p(x) +  \chi(\norm{u}_\infty)
    \end{align*}
    for all $x \in X$, $u \in U$, where we used \eqref{ineq:LMI1} in the fourth step.

    For the jump dynamics of $\widehat\Sigma$, we have
    \begin{align*}
        V_p(x) &= x^T M_p x \\
        &= (J_q x^- + H_q u^-)^T M_p (J_q x^- + H_q u^-) \\
        &= (x^-)^T J_q^T M_p J_q x^- + (u^-)^T H_q^T M_p J_q x^- + (x^-)^T J_q^T M_p H_q u^- + (u^-)^T H_q^T M_p H_q u^- \\
        &\leq (x^-)^T M_q x^- + (u^-)^T Q_q u^- \\
        &\leq \mu_q V_q(x^-) + \chi(\norm{u}_\infty)
    \end{align*}
    for all $x^- \in X$, $u^- \in U$, such that $x = J_q x^- + H_q u^-$. Note that we applied \eqref{ineq:LMI2} in the fourth step.
    
    Therefore, $V_\sigma$ defines a non-decreasing Lyapunov function in dissipation form with rates $\varphi_p(s) = \eta_p s$ and $\psi_p(s) = \mu_p s$ for $s \in \R^+_0$. Note that \eqref{ineq:LMI3} and \eqref{ineq:LMI4} correspond to \eqref{ineq:dwellTimeConditionSFUJ} and \eqref{ineq:dwellTimeConditionUFSJ} (cf. Remark \ref{rem:linearDwellTimeCond}). Then, by Lemma \ref{lem:dissiptoimplicatLF} and Corollary \ref{cor:robustISS}, $\widehat\Sigma$ is ISS.
\end{proof}

\section{Conclusion}
\label{sec:conclusion}
We provided necessary and sufficient ISS conditions for impulsive switched systems that have modes with both stable and unstable flows. To achieve this, we used time-varying ISS-Lyapunov functions with nonlinear rate functions, along with  MDADT and MDALT conditions. We also presented a method for constructing decreasing ISS-Lyapunov functions from non-decreasing ones, which is an important and useful result in its own right. Additionally, we provided a method to guarantee ISS for a particular class of impulsive switched systems with unknown switching signals. 
\section{Appendix}


\subsection{Technical Lemmas}
\begin{lemma}\label{lem:LipschitzContinuousFlow}
	Consider system $\Sigma$ for a given switching signal $\sigma$. Let $f_p$ be locally Lipschitz continuous in the second argument, uniformly for all $t \in I$ and $p \in \mathcal M$, and $u \in B_{U_c}(D)$ for some constant $D > 0$.
	Let the function $g_p$ be locally Lipschitz continuous in the second variable, uniformly with respect to  $t \in I$ and $p \in \mathcal M$, and $u,v \in B_{U_c}(D)$. Then, on compact time intervals, uniform for all switching sequences, the solutions of $\Sigma$ are locally Lipschitz continuous with respect to the initial values, i.e., for all $C > 0$, $x_0,y_0 \in B_X(C)$, and for $\tau > t_0$, there exists a constant $L = L(C, D, \tau, \sigma)$, such that for $x(t) = x(t;t_0,x_0,u,\sigma)$ and $y(t) = x(t;t_0,y_0,v,\sigma)$, it holds that
    \begin{align*}
		&\max_{t \in [t_0,\tau]}\norm{x(t) - y(t)}_X \leq L(C, D, \tau, \sigma) \norm{x_0 - y_0}_X.
	\end{align*}
\end{lemma}
\begin{proof}
	Let $C > 0$, $\tau > t_0$, $u \in B_{U_c}(D)$, and $x_0, y_0 \in B_X(C)$. Let us denote the Lipschitz constant for $f_p = f_p(t,x,u)$ and $g_p = g_p(t,x,u)$, for $\norm{x} \leq C$ and $\norm{u} \leq C$, by $L_f = L_f(C)$ and $L_g(C)$, respectively.
	
    Let $t \in [t_0, t_1)$ be such that $t \leq \tau$. Then, we have
	\begin{align*}
		x(t) = x_0 + \int_{t_0}^t f_{\sigma(s)}\!\para{s, x(s), u(s)} \,\diff s.
	\end{align*}
	Let $x_i^- := x(t_i^-)$ and $x_i := x(t_i) = g_{\sigma(t_i^-)}(t_i^-,x_i^-,u(t_i^-))$ be such that
	\begin{align*}
		x(t) = x_i + \int_{t_i}^t f_{\sigma(s)}\!\para{s, x(s), u(s)} \,\diff s
	\end{align*}
	holds true for all $t \in [t_i, t_{i + 1})$, $t \leq \tau$.
	From this, for any two solutions $x$ and $y$ of $\Sigma$ with initial values $x_0$ and $y_0$, respectively, it follows that
	\begin{align*}
		\norm{x(t) - y(t)}_X \leq \norm{x_i - y_i}_X + \int_{t_i}^t L_f\!\para{K(C, D, \tau, \sigma)} \norm{x(s) - y(s)}_X \,\diff s,
	\end{align*}
	where $t \in [t_i, t_{i + 1})$, $t \leq \tau$, $y_i^- := y(t_i^-)$, $y_i := y(t_i) = g_{\sigma(t_i^-)}(t_i^-,y_i^-,v(t_i^-))$, and $K = K(C, D, \tau, \sigma)$, as introduced in~\eqref{eq:forwardCompleteness}. Using Gronwall's inequality, it follows that
	\begin{align*}
		\norm{x(t) - y(t)}_X \leq \norm{x_i - y_i}_X e^{L_f\para{K}(t - t_i)}
	\end{align*}
	for all $t \in [t_i,t_{i+1}), t \leq \tau$. Moreover,
	\begin{align*}
		\norm{x(t_{i+1}) - y(t_{i+1})}_X 
		&\leq L_g\!\para{K} \norm{x(t_{i + 1}^-) - y(t_{i + 1}^-)}_X \leq L_g\!\para{K}  \norm{x_i - y_i}_X \!e^{L_f\para{K}(t_{i+1} - t_i)},
	\end{align*}
	where $L_g$ is the Lipschitz constant.
	By induction, we obtain
	\begin{align*}
		&\norm{x(t) - y(t)}_X \leq L_g^n\!\para{K} e^{\para{L_f\!\para{K}}(t - t_0)} \norm{x_0 - y_0}_X 
	\end{align*}
	for $t \in [t_n, t_{n+1})$, $ t \leq \tau$ and $n \in \N_0$.

    As the switching sequence does not exhibit Zeno behavior, it follows that
    \begin{align*}
		\max_{t \in [t_0,\tau]}\norm{x(t) - y(t)}_X \leq \norm{x_0 - y_0}_X \max\limits_{t \in [t_0, \tau]} \! \braces{L_g^{N(t_0,t)}\!\para{K} e^{L_f\!\para{K}(t - t_0)}}.
	\end{align*}
    By the extreme value theorem, the bound 
    \begin{align*}
        L(C, D, \tau, \sigma) := \max\limits_{t \in [t_0, \tau]} \! \braces{L_g^{N(t_0,t)}\!\para{K(C, D, \tau, \sigma)} e^{L_f\!\para{K(C, D, \tau, \sigma)}(t - t_0)}}
    \end{align*}
    exists and gives us a Lipschitz constant for the solutions of $\Sigma$ with respect to the initial conditions.
\end{proof}

\begin{lemma}\label{lem:dissiptoimplicatLF}
    If a system $\Sigma$ has a nondecreasing ISS-Lyapunov function in dissipation form with linear rates $\widetilde \varphi_p(s) = \widetilde \eta_p s$, $\widetilde \psi_p(s) = \widetilde \mu_p s$, $s \in \R^+_0$ and $\widetilde \chi_p\in\KK_\infty$, then it has a nondecreasing ISS-Lyapunov function (in implication form) with linear rate function. 
    \\
    We can choose $\varphi_p(s) := \eta_p s$, $\psi_p(s) = \mu_p s$ for $s \in \R^+_0$, where $\eta_p = \frac{1 - \delta}{1 - \frac{3}{4}\delta}\widetilde \eta_p$ in case $p \in \mathcal S$, i.e., $\widetilde \eta_p < 0$, and $\eta_p = \frac{1 + \delta}{1 + \frac{3}{4}\delta}\widetilde \eta_p$ else, as well as $\mu_p = e^{-\frac{\delta \tau_p \eta_p}{4}} \widetilde \mu_p$ and $\chi_p(s) := \min\!\braces{\frac{1}{\eta_p - \widetilde \eta_p}, \frac{1}{\mu_p - \widetilde \mu_p}}\cdot \widetilde \chi_p (s)$.
\end{lemma}
\begin{proof}
     With the choices of $\varphi_p, \psi, \chi, \eta$, and $\mu$ from above follows that $V_p(x) \geq \chi_p(\norm u_\infty)$. Then, it follows for each $p \in \mathcal M$ that
     \begin{align*}
         \dot V_p(x) &\leq \widetilde \varphi_p(V_p(x)) + \widetilde \chi_p (\norm u_\infty)
         \leq \widetilde \eta_p V_p(x) + (\eta_p - \widetilde \eta_p) V_p(x) = \eta_p V_p(x) = \varphi_p(V_p(x)), \\
         V_p(x) &\leq \widetilde \psi_p(V_p(x)) + \widetilde \chi_p (\norm u_\infty)
         \leq \widetilde \mu_p V_p(x) + (\mu_p - \widetilde \mu_p) V_p(x) = \mu_p V_p(x) = \psi_p(V_p(x)).
     \end{align*}
     Furthermore, for $p \in \mathcal S$ and $\frac{\ln \tilde \mu_p}{\abs{\tilde \eta_p}} \leq \tau_p (1 - \delta)$, it holds that
     \begin{align*}
         \frac{\ln \mu_p}{\abs{\eta_p}} 
         &= \frac{\ln \!\para{e^\frac{\delta \tau_p \abs{\eta_p}}{4} \widetilde \mu_p}}{\abs{\eta_p}} 
         = \frac{\frac{\delta \tau_p \abs{\eta_p}}{4}}{\abs{\eta_p}} + \frac{\ln  \widetilde \mu_p}{\frac{1 - \delta}{1 - \frac{3}{4}\delta}\abs{\widetilde \eta_p}} = \frac{\delta \tau_p}{4}
         + \tau_p\para{1 - \frac{3}{4}\delta}
         \leq \tau_p \para{1 - \delta'},
     \end{align*}
     where $\delta' := \frac{\delta}{2}$. From this, it follows immediately that $V_p$ is a nondecreasing ISS-Lyapunov function in implication form. A similar procedure achieves the result for $p \in \mathcal U$.
\end{proof}

\bibliographystyle{IEEEtran}
\bibliography{main}

\begin{IEEEbiography}[{\includegraphics[width=1in,height=1.25in,clip,keepaspectratio]{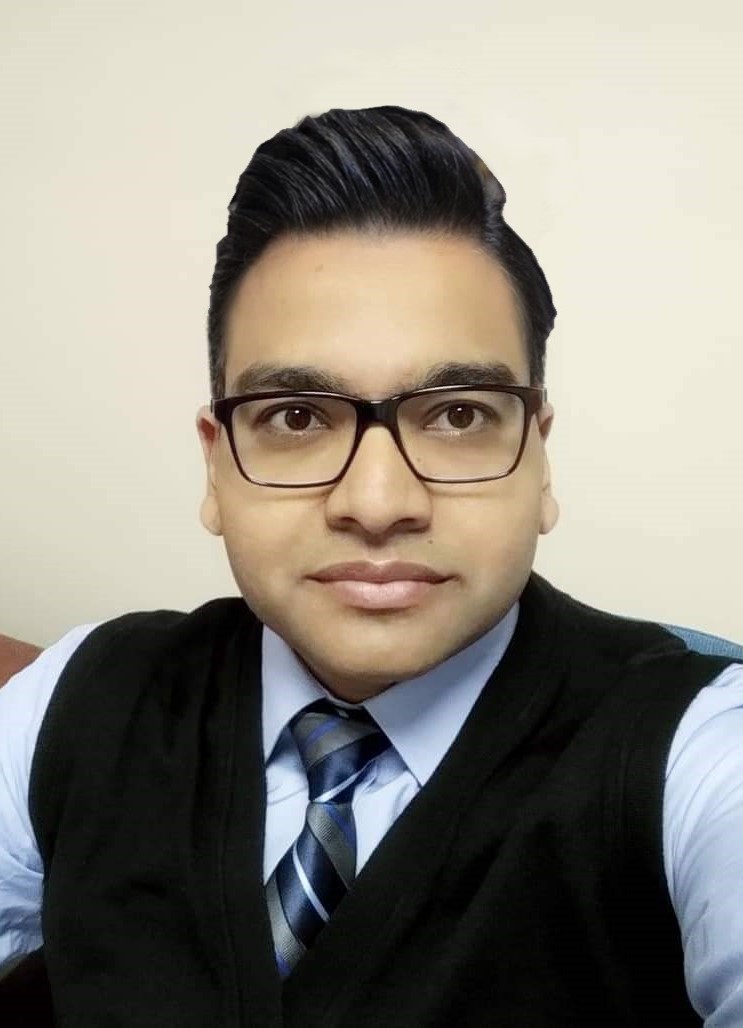}}]{Saeed Ahmed} is an Assistant Professor of Systems and Control at the University of Groningen, where he is affiliated with the Engineering and Technology Institute Groningen and the Jan C. Willems Center for Systems and Control. Prior to joining this position, he held postdoctoral positions at the University of Groningen with Jacquelien Scherpen and at the Technical University of Kaiserslautern (now RPTU), Germany. He completed his Ph.D. at Bilkent University, Turkey. During his Ph.D., he was a visiting scholar at CentraleSupélec, France. His Ph.D. was supervised by  Frederic Mazenc and Hitay Ozbay. He also collaborated with Micheal Malisoff during his Ph.D. His research interests span various topics in systems and control engineering. From a theoretical point of view, he is interested in stability and control, online optimization, observer design, nonlinear and hybrid (switched and impulsive) systems, dissipativity and passivity analysis, robust control, LPV systems, and time-delay systems. From an application point of view, he is interested in designing intelligent control algorithms for autonomous vehicles and district heating systems.   He received the best presentation award in the Control/Robotics/Communications/Network category at the IEEE Graduate Research Conference 2018 held in Bilkent University, Turkey. He is an associate editor of Systems and Control Letters and a member of the IFAC Technical Committees on Networked Systems, Non-linear Control Systems, and Distributed Parameter Systems. 
\end{IEEEbiography}

\begin{IEEEbiography}[{\includegraphics[width=1in,height=1.25in,clip,keepaspectratio]{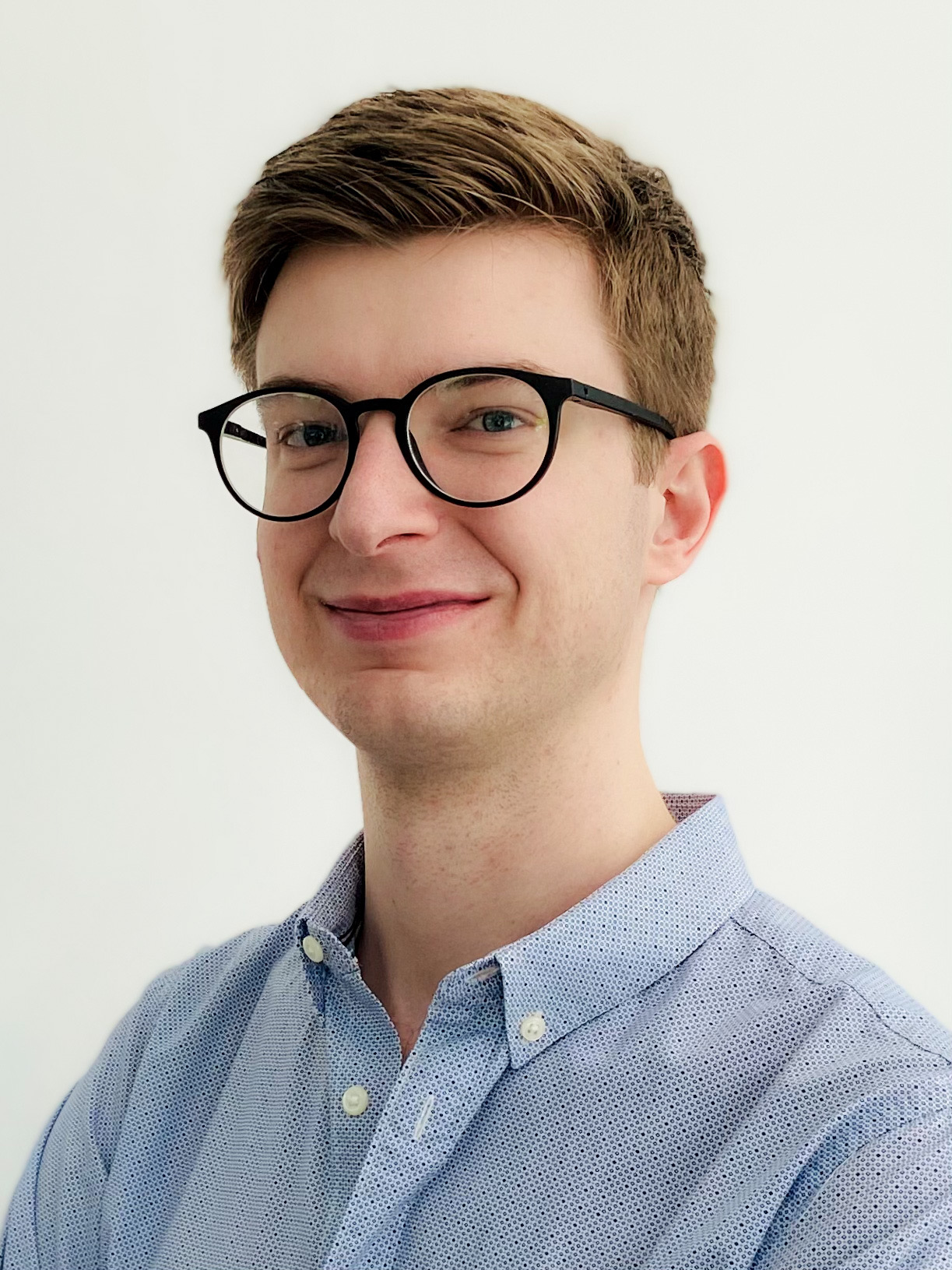}}]{Patrick Bachmann} was born in Speyer, Germany, in 1992. He received his Bachelor's degree in business mathematics from the University of Mannheim, Germany, in 2015 and his Master's degree in Mathematics from Karlsruhe Institute of Technology in 2018. From 2018 to 2021, he was a research assistant at the Faculty of Mechanical and Process Engineering of the Technical University of Kaiserslautern, Germany. Since 2021, he is a teaching assistant at the Institute of Mathematics at the University of Würzburg in Germany where he is also pursuing his Ph.D. degree in Mathematics under the supervision of Sergey Dashkovskiy and Andrii Mironchenko. His research interests include impulsive systems, stability and control theory, Lyapunov functions and infinite dimensional systems.
\end{IEEEbiography}

\begin{IEEEbiography}[{\includegraphics[width=1in,height=1.25in,clip,keepaspectratio]{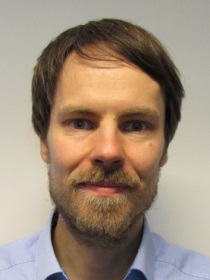}}]{Stephan Trenn} received his Ph.D. (Dr. rer. nat.) within
the field of differential algebraic systems and distribution
theory at the Ilmenau University of Technology,
Germany, in 2009. Afterwards, he held Postdoc positions
at the University of Illinois at Urbana-Champaign,
USA (2009–2010) and at the University of Würzburg,
Germany (2010–2011). After being an Assistant Professor
(Juniorprofessor) at the University of Kaiserslautern,
Germany, he became Associate Professor for Systems
and Control at the University of Groningen, Netherlands,
in 2017. He is an Associate Editor for the journals
Systems and Control Letters, Nonlinear Analysis: Hybrid Systems, IEEE Control
Systems Letters, and DAE Panel.
\end{IEEEbiography}

\end{document}